\documentclass[nofootinbib,pra,longbibliography,twocolumn,showpacs,superscriptaddress,notitlepage]{revtex4-2}
\usepackage{amsmath}
\usepackage{amssymb,bm}
\usepackage{amsthm}
\usepackage{qcircuit}
\usepackage{bbold}
\newtheorem{innercustomgeneric}{\customgenericname}
\providecommand{\customgenericname}{}
\newcommand{\newcustomtheorem}[2]{%
  \newenvironment{#1}[1]
  {%
   \renewcommand\customgenericname{#2}%
   \renewcommand\theinnercustomgeneric{##1}%
   \innercustomgeneric
  }
  {\endinnercustomgeneric}
}

\newcustomtheorem{customthm}{Theorem}
\newcustomtheorem{definitionBob}{Definition}

\usepackage{color,dsfont} 
\usepackage{graphicx}
\usepackage{subcaption}
\usepackage{ragged2e}
\DeclareCaptionJustification{justified}{\justifying}
\usepackage[colorlinks=true, hyperindex, breaklinks, linkcolor=blue, urlcolor=blue, citecolor=blue]{hyperref} 
\usepackage[normalem]{ulem}
\usepackage[capitalise]{cleveref}
\usepackage{mathrsfs}

\usepackage{mathtools}
\usepackage{float}
\usepackage{verbatim}
\usepackage{latexsym}
\usepackage{setspace}
\usepackage{amsfonts}
\usepackage{stmaryrd}
\usepackage{xcolor}
\usepackage{enumitem}
\usepackage{hhline}
\usepackage[normalem]{ulem}

\newtheorem{theorem}{Theorem} \newtheorem{lemma}{Lemma}
\newtheorem{corollary}{Corollary} 
 \newtheorem{definition}{Definition}
\newtheorem{proposition}{Proposition}

\newtheorem{algorithm}{Algorithm}

\def\ket#1{\left| #1 \right\rangle}

\newcommand{\codepar}[1]{\ensuremath{[\![#1]\!]}}

\usepackage{amsbsy}
\newcommand*\Bell{\ensuremath{\boldsymbol\ell}}

\crefname{equation}{Eq.\!}{Eqs.\!}
\crefname{figure}{Fig.\!}{Figs.\!}
\usepackage{multirow}
\mathchardef\mhyphen="2D

\begin{document}

\title{Clifford gates with logical transversality for self-dual CSS codes}

\author{Theerapat Tansuwannont}
\email{t.tansuwannont@phys.s.u-tokyo.ac.jp}
\affiliation{Department of Physics, Graduate School of Science, The University of Tokyo, 7-3-1 Hongo, Bunkyo-ku, Tokyo, 113-0033, Japan}
\affiliation{Center for Quantum Information and Quantum Biology, The University of Osaka, Toyonaka, Osaka 560-0043, Japan}

\author{Yugo Takada}
\email{u751105k@ecs.osaka-u.ac.jp}
\affiliation{Center for Quantum Information and Quantum Biology, The University of Osaka, Toyonaka, Osaka 560-0043, Japan}
\affiliation{Graduate School of Engineering Science, The University of Osaka, Toyonaka, Osaka 560-8531, Japan}

\author{Keisuke Fujii}
\email{fujii.keisuke.es@osaka-u.ac.jp}
\affiliation{Center for Quantum Information and Quantum Biology, The University of Osaka, Toyonaka, Osaka 560-0043, Japan}
\affiliation{Graduate School of Engineering Science, The University of Osaka, Toyonaka, Osaka 560-8531, Japan}
\affiliation{RIKEN Center for Quantum Computing, Hirosawa 2-1, Wako, Saitama 351-0198, Japan}

\begin{abstract}
Quantum error-correcting codes with high encoding rate are good candidates for large-scale quantum computers as they use physical qubits more efficiently than codes of the same distance that encode only a few logical qubits. Some logical gate of a high-rate code can be fault-tolerantly implemented using transversal physical gates, but its logical operation may depend on the choice of a symplectic basis that defines logical Pauli operators of the code. In this work, we focus on \codepar{n,k,d} self-dual Calderbank-Shor-Steane (CSS) codes with $k \geq 1$ and prove necessary and sufficient conditions for the code to have a symplectic basis such that (1) transversal logical Hadamard gates $\bigotimes_{j=1}^{k} \bar{H}_j$ can be implemented by transversal physical Hadamard gates $\bigotimes_{i=1}^{n} H_i$, and (2) for any $(a_1,\dots,a_k)\in\{-1,1\}^k$, transversal logical phase gates $\bigotimes_{j=1}^{k} \bar{S}_j^{a_j}$ can be implemented by transversal physical phase gates $\bigotimes_{i=1}^{n} S_i^{b_i}$ for some $(b_1,\dots,b_n)\in\{-1,1\}^n$. Self-dual CSS codes satisfying the conditions include any codes with odd $n$. We also generalize the idea to concatenated self-dual CSS codes and show that certain logical Clifford gates have multiple transversal implementations, each by logical gates at a different level of concatenation. Several applications of our results for fault-tolerant quantum computation with low overhead are also provided.
\end{abstract}

\maketitle

\section{Introduction}

To build a reliable quantum computer, one has to ensure that quantum information during the computation is not affected by unwanted interactions with the environment. Quantum error correction (QEC) is a technique to reduce the error rate by encoding some logical qubits into a larger number of physical qubits using a quantum error correcting code (QECC). To attain lower logical error rates by increasing the number of physical qubits, a family of QECCs that can attain high code distance such as concatenated codes \cite{KL96}, topological codes \cite{Kitaev97,BK98,BM06}, or quantum low-density parity check (qLDPC) codes \cite{MMM04} is required. Fault-tolerant error correction (FTEC) schemes ensure that a few faults that can happen during the QEC implementation cannot propagate and cause uncorrectable errors. To process the logical quantum information in a fault-tolerant manner, schemes for fault-tolerant quantum computation (FTQC) are needed. With these fault-tolerant tools and a proper family of codes, one can simulate any quantum circuit with an arbitrarily low logical error rate if the physical error rate is below some scheme-dependent threshold value \cite{Shor96,AB08,Kitaev97,KLZ96,Preskill98,TB05,ND05,AL06,AGP06}.

Achieving fault tolerance generally requires large space overhead (e.g. ancilla qubits) and large time overhead (e.g., quantum gates). One way to perform a logical operation on the code space with constant time overhead and without additional ancilla qubits is to use transversal gates; On a certain QECC, a one-qubit (or a two-qubit) logical gate can be implemented by transversally applying physical single-qubit gates on all physical qubits in the code block (or transversally applying physical two-qubit gates on all pairs of physical qubits between two code blocks). Transversal gates implemented by either single-qubit or two-qubit physical gates are naturally fault tolerant since they do not spread errors inside each code block, and even if errors propagate to another code block or some faults occur during their implementation, a subsequent QEC scheme can still correct errors as intended \cite{Gottesman97,AGP06}.

The set of available transversal gates depends on the QECC being used, and any QECC cannot have the gate set for universal quantum computation consisting of only transversal gates \cite{ZCC11,CCCZC08,EK09}. For some code that encodes a single logical qubit in a code block such as the Steane code \cite{Steane96} or color codes \cite{BM06} with certain boundary conditions, logical Hadamard ($H$), phase ($S$), and controlled-NOT (CNOT) gates can be implemented transversally, generating the full logical Clifford group. For such codes, a logical non-Clifford gate can be implemented through techniques such as magic state distillation and injection \cite{KLZ97,BK05} or code switching \cite{PR13,ADP14}, completing the universal gate set. However, for a QECC that encodes multiple logical qubits in a code block, it is possible that some logical Clifford gates cannot be achieved by transversal gates only, and additional techniques are required to implement such logical gates in order to achieve the full logical Clifford group. 

Recently, stabilizer codes with high encoding rate have gained a lot of attention as they can ease the process of scaling up a quantum computer. There are recent breakthroughs in good qLDPC codes \cite{BE21,HHJO21,EKZ22,LZ22,PK22a,PK22b,DHLV23}. The weights of the stabilizer generators of a qLDPC code in each family are upper-bounded by some constant independent of the block length, so scaling up the code size will not affect the overhead required for FTEC schemes in which the construction of a circuit to measure each stabilizer generator depends on the weight of the operator. Examples of such FTEC schemes are Shor \cite{Shor96,DA07} and flag \cite{CR18,CR20} schemes. The development of techniques for implementing logical gates and performing FTQC on qLDPC codes is still an active research area.

There are also several families of stabilizer codes with high encoding rate which are not qLDPC codes. Examples of such families are quantum Hamming codes \cite{Steane96_qHamming}, quantum Reed--Muller codes \cite{ADP14}, and quantum Bose--Chaudhuri--Hocquenghem (qBCH) codes \cite{GB99}. There are no constraints on the weights of the stabilizer generators of such codes, so FTEC schemes that depend on the stabilizer weights may not scale well when the code distance increases. Nevertheless, the scaling issue of a FTEC scheme can be eased if a scheme that does not depend on the stabilizer weights such as Steane \cite{Steane97,Steane04} or Knill \cite{Knill05} FTEC scheme is used; Steane and Knill schemes which are applicable to Calderbank-Shor-Steane (CSS) codes \cite{CS96,Steane96} and stabilizer codes \cite{Gottesman97,Gottesman96}, respectively, require two blocks of ancilla states for measuring all stabilizer generators simultaneously. FTQC on CSS codes encoding multiple logical qubits can be performed by the teleportation-based FTQC schemes developed in Refs. \cite{BZHJL15,ZLB18,ZLBK20}; the FTQC scheme for logical $H$, $S$, CNOT, or SWAP gate \cite{BZHJL15} (or the FTQC scheme for any Clifford circuit \cite{ZLBK20}) requires constant time overhead given that certain ancilla states can be prepared fault-tolerantly and efficiently \cite{ZLB18,ZLBK20}. Recently, Yamasaki and Koashi \cite{YK24} have developed an FTQC scheme for concatenated codes obtained from concatenating quantum Hamming codes with growing size, achieving constant space overhead and quasi-polylogarithmic time overhead. Their construction utilizes Knill FTEC scheme and relies on the assumption that for any quantum Hamming code, transversally applying physical Hadamard gates induces logical Hadamard gates on all logical qubits in a code block.

It should be noted that for a stabilizer code that encodes multiple logical qubits, there is more than one way to define a symplectic basis, i.e., the basis for logical Pauli operators. The logical operation corresponding to a given physical operation depends on the choice of a symplectic basis, thus changing a symplectic basis could result in a different logical operation on the logical qubits. It was not shown explicitly in Ref. \cite{YK24} on which symplectic basis the concatenated quantum Hamming code is operating on, so the logical operation of transversal physical Hadamard gates is unclear. 

In this work, we focus on a family of \codepar{n,k,d} self-dual CSS codes with $k \geq 1$ (which includes quantum Hamming codes \cite{Steane96_qHamming} and binary qBCH codes \cite{GB99}). For such codes, we aim to find a symplectic basis compatible with multilevel transversal Clifford operations, or a \emph{compatible symplectic basis}, which is a symplectic basis such that logical Hadamard gates on all logical qubits can be implemented by transversal physical Hadamard gates, and logical $S$ and $S^\dagger$ gates on all logical qubits can be implemented by transversal physical $S$ and $S^\dagger$ gates (\cref{def:compatible_basis}). Since for any CSS code, transversal logical CNOT gates applied to all pairs of logical qubits between two code blocks can be implemented by transversal physical CNOT gates, a self-dual CSS code with a compatible symplectic basis admits block transversal implementation of the full Clifford group (where the same logical Clifford gate is applied to all logical qubits in the same code block). 

Our main results are the following: (1) We prove necessary and sufficient conditions for any self-dual CSS code with $k \geq 1$ to have a compatible symplectic basis in \cref{thm:main1}. Codes that satisfy the necessary and sufficient conditions include any \codepar{n,k,d} self-dual CSS codes with $k \geq 1$ and odd $n$ (\cref{cor:odd_n}). With some additional addressable logical Clifford gates, the full logical Clifford group defined on all logical qubits across all code blocks can be achieved (\cref{prop:full_Clifford}). Such addressable gates can be implemented by the FTQC scheme developed by Brun et al. \cite{BZHJL15} which utilizes Steane fault-tolerant measurement (FTM) scheme (see \cref{app:gate_TP}). (2) We extend our results to an \codepar{N,K,D} concatenated code obtained from concatenating self-dual CSS codes that satisfy \cref{thm:main1}. We define logical gates and transversality at each level of concatenation, then prove in \cref{thm:main2,thm:main3} that some logical gates at a certain level exhibit transversality at multiple levels of concatenation. A logical gate with this property has multiple transversal implementations, each by logical gates at a different level. Consequently, we can show that transversal logical Hadamard gates (transversal logical $S$ and $S^\dagger$ gates) applied to all logical qubits at the top level of concatenation can be implemented by transversal physical Hadamard gates (transversal physical $S$ and $S^\dagger$ gates; see \cref{cor:trans_logical_physical}). (3) We provide a procedure to construct a compatible symplectic basis of any self-dual CSS code if the basis exists in \cref{app:python}. The procedure can also construct phase-type logical-level transversal gates of the form $\bigotimes_{j=1}^{k} \bar{S}_j^{a_j}$ for any $(a_1,\dots,a_k)\in\{-1,1\}^k$ as a combination of physical $S$ and $S^\dagger$ gates. A Python implementation of our procedure is available at \url{https://github.com/yugotakada/mlvtrans}.

We also propose several applications of compatible symplectic bases and Clifford gates with multilevel transversality: (1) Our results provide an explicit construction of symplectic bases for quantum Hamming codes which complements the constant overhead FTQC scheme by Yamasaki and Koashi \cite{YK24}. Furthermore, our results suggest possibilities to extend their scheme to concatenated codes obtained from concatenating high-rate self-dual CSS codes, as well as possibilities to optimize the scheme due to the fact that any transversal logical $S$ and $S^\dagger$ gates can be implemented by some transversal physical $S$ and $S^\dagger$ gates. (2) For concatenated self-dual CSS codes, we propose several techniques to convert certain logical gates to alternative sets of gates which have lower implementation complexity and are better at error confinement. We also demonstrate that a product of two logical gates which are transversal at some level $m$ could result in a logical gate which is transversal at level $m' < m$. These conversion techniques could simplify gates in fault-tolerant protocols for a concatenated self-dual CSS code especially at the locations where logical gates from different levels are implemented consecutively. (3) We show that for a self-dual CSS code that has a compatible symplectic basis, the types of ancilla states required for the teleportation-based FTQC scheme in Ref. \cite{BZHJL15} can be substantially simplified due to the fact that several transversal logical gates can be implemented by transversal physical gates. This simplification could play an important role in an architectural design of quantum computers that utilize teleportation-based FTQC schemes.

We note that rather than an immediate demonstration of practical advantages for universal fault-tolerant computation, our work aims to provide a structural characterization that facilitates the identification and analysis of codes admitting transversal logical Clifford operations. We believe that by using a compatible symplectic basis, one could potentially reduce the implementation overhead of certain quantum algorithms in the logical level.
In particular, any Clifford circuit can be decomposed into layers of $H$, $S$, CNOT gates \cite{AG04,MR18}. If a significant fraction of these layers involve gates on many logical qubits, then transversal implementations of logical $H$ or $S$ gates combined with individually implemented logical gates may lead to lower implementation overhead. 
Large layers of Hadamard operations arise in many well-known algorithms such as Bernstein-Vazirani algorithm \cite{BV97}, Grover's algorithm \cite{Grover96}, and quantum Fourier transform-based algorithms including Shor's algorithm \cite{Shor94}, so codes with compatible symplectic bases could be useful for the implementation of such algorithms.

This paper is organized as follows. In \cref{sec:related_works}, we discuss several works related to transversal gates and other techniques for implementing logical gates on stabilizer codes. In \cref{sec:definitions}, we provide definitions related to fault tolerance and transversal gates, formally define a compatible symplectic basis, and discuss additional addressable logical Clifford gates required to achieve the full logical Clifford group if the code has a compatible symplectic basis. We then provide a construction of a compatible symplectic basis of the \codepar{15,7,3} quantum Hamming code in \cref{sec:qHamming} as a motivating example. In \cref{sec:suf_nec_cond}, we prove necessary and sufficient conditions for existence of a compatible symplectic basis of any \codepar{n,k,d} self-dual CSS code with $k \geq 1$, as well as related algorithm, corollary, and examples. Later in \cref{sec:multilevel_trans}, we extend our results to concatenated codes by defining logical gates and transversality at each level of concatenation, and prove that some Clifford gates exhibit multilevel transversality. We propose several applications of compatible symplectic bases and Clifford gates with multilevel transversality in \cref{sec:applications}, then discuss and conclude our work in \cref{sec:conclusion}.

\section{Related works} \label{sec:related_works}

Calderbank et al. \cite{CRSS98} showed that any stabilizer code can be mapped to a classical code over $GF(4)$, and all logical operators of the stabilizer code can be found by calculating the automorphism group of the corresponding classical code. A method to synthesize a physical Clifford circuit for any logical Clifford gate was proposed by Rengaswamy et al. \cite{RCKP20}; note that the corresponding circuit may or may not be fault tolerant. A recent work by Sayginel et al. \cite{SKWRB24} proposed a method to find all logical Clifford gates that satisfy certain physical constraints, including transversal and fold-transversal gates, for any stabilizer code. It should be noted that the method in Ref. \cite{SKWRB24} constructs a symplectic basis (i.e., logical X and logical Z operators) by representing a stabilizer code in the standard form \cite{NC00}, and the logical operation of each operator found by their method is defined by this choice of symplectic basis. One may want to transform an available logical operator in one symplectic basis to a desired logical operator in another symplectic basis. However, we conjecture that the number of possible symplectic bases of each code grows exponentially in the number of logical qubits, so checking whether the transformation is possible or not is generally hard for a code with high encoding rate. 

Jain and Albert \cite{JA24} studied high-distance \codepar{n,1,d} CSS codes that allow transversal Clifford or transversal $T$ gates. In their work, classical quadratic-residue codes are used to constructed self-dual doubly even CSS codes, in which any logical Clifford gate can be implemented transversally. Weak triply even CSS codes that admit transversal implementation of logical $T$ gate are constructed from the doubly even CSS codes through a doubling procedure \cite{BC15}.

Several works have considered families of codes that admit transversal non-Clifford gates. Rengaswamy et al. \cite{RCNP20} proved necessary and sufficient conditions for stabilizer codes to admit an implementation of logical $T$ gates on all logical qubits by transversal physical $T$ gates. The necessary and sufficient conditions were generalized by Hu et al. \cite{HLC22} to transversal $Z$ rotation through $\pi / 2^l$ and other diagonal gates, while allowing additional Pauli correction to preserve the code space. There are works that construct quantum codes admitting transversal implementation of a targeted diagonal gate \cite{HLC21}, and code families in which code distance $d$ and the number of logical qubits $k$ grow linearly in block length $n$ while admitting transversal controlled-controlled-$Z$ gates \cite{GG24}. Algorithms in Ref. \cite{WBB22,WQB23} provide ways to find all logical diagonal gates that can be implemented transversally on any CSS code. For various families of topological and qLDPC codes, constructions of transversal non-Clifford gates were proposed \cite{Brown24,ZSPCB24,GG24,Lin24,GL24}. It has been observed by several works that there is a close connection between transversal gates on qLDPC codes and the cup product in algebraic topology \cite{BDET24,ZSPCB24,Lin24,GL24,Zhu25}.

The full logical Clifford group of a stabilizer code encoding multiple logical qubits is generally hard to achieve with transversal gates only. Fold-transversal gates \cite{Moussa16,BB24} which consist of single-qubit and two-qubit physical gates provide an alternative way to implement a logical gate a code block. This technique allows a richer set of logical gates that can be implemented with constant time overhead and without additional qubits. However, depending on the noise model, fold-transversal gates may be less fault tolerant than strictly transversal gates (i.e., the effective code distance may decrease) due to the fact that each two-qubit physical gate can affect multiple physical qubits in a single code block. Fold-transversal gates for various families of qLDPC codes have been recently studied \cite{BB24,TB24,GV24,ES25}, and available fold-transversal gates for any code can be found using the method proposed in Ref. \cite{SKWRB24}.

In case that a desired logical Clifford gate cannot be implemented by transversal or fold-transversal gates, other techniques can be used. Any logical Clifford gate can be implemented by logical gate teleportation, which involves fault-tolerant logical Bell measurement and fault-tolerant preparation of some stabilizer states as ancilla states. One drawback of this scheme is the number of required measurements is large, and different logical Clifford gates require different ancilla states. Brun et al. \cite{BZHJL15} have proposed another teleportation-based FTQC scheme that utilizes Steane FTM scheme \cite{Steane99}, which can jointly measure logical Pauli operators on a code block or between code blocks while requiring simpler ancilla states. Elementary logical Clifford gates such as logical $H$, $S$, CNOT, and SWAP gates can be implemented by logical Pauli measurements. An efficient fault-tolerant protocol for preparing the ancilla states which are required in Ref. \cite{BZHJL15} is proposed in Ref. \cite{ZLB18}. A teleportation-based FTQC scheme that can implemented any logical Clifford circuit in constant time was proposed in Ref. \cite{ZLBK20}. Note that the last scheme transfers the complexity of logical Clifford circuits to the complexity of preparation of required ancilla states.

On some families of topological codes, some logical two-qubit gate inside a code block can be implemented by Dehn twists \cite{BVCKT17}, a technique that applies a series of physical two-qubit gates on a code block. To make the implementation fault tolerant, multiple rounds of error correction must be applied after each two-qubit gate. It has been shown on 2D toric codes that a set of logical Clifford gates required to generate the full logical Clifford group can be implemented in constant time by combining transversal and fold-transversal gates, Dehn twists, and single-shot logical Pauli measurements \cite{GV24}. The concept of Dehn twists has recently been generalized to broader families of qLDPC codes \cite{TB24}.

\section{Transversal gates and compatible symplectic bases} \label{sec:definitions}

In this section, we provide definitions related to fault tolerance and transversal gates, state the scope of this work, and formally define a compatible symplectic basis. Since the full logical Clifford group of a high-rate code generally cannot be achieved by transversal gates only, we also provide a proposition stating additional addressable logical Clifford gates which are required to achieve the full logical Clifford group.  These definitions and proposition will be later used in the constructions throughout this work.

An \codepar{n,k,d} \emph{stabilizer code} \cite{Gottesman97,Gottesman96} is a quantum code that encodes $k$ logical qubits using $n$ physical qubits and can correct errors on up to $t = \lfloor (d-1)/2 \rfloor$ physical qubits, where $d$ is the code distance. Any stabilizer code can be described by its corresponding \emph{stabilizer group}, an Abelian group generated by $n-k$ commuting Pauli operators that does not include $-I^{\otimes n}$. The code space of a stabilizer code is the $+1$ simultaneous eigenspace of all \emph{stabilizers}, the elements of the stabilizer group. A \emph{Calderbank-Shor-Steane (CSS) code} \cite{CS96,Steane96} is a stabilizer code in which generators of its stabilizer group can be chosen to be purely $X$-type or purely $Z$-type. Any CSS code can be described by parity check matrices $H_x \in \mathbb{Z}_2^{r_x \times n}$ and $H_z \in \mathbb{Z}_2^{r_z \times n}$ of two classical binary linear codes $\mathcal{D}_x$ and $\mathcal{D}_z$ satisfying $\mathcal{D}_x^\perp \subseteq \mathcal{D}_z$ (where $r_x+r_z = n-k$). Each row of $H_x$ ($H_z$) corresponds to an $X$-type ($Z$-type) stabilizer generator, where $0$ and $1$ correspond to $I$ and $X$ ($Z$) tensor factors of the generator. In this work, we focus on a \emph{self-dual CSS code}, a CSS code in which $H_x = H_z$\footnote{In some literature that describes CSS codes in the language of homology, the dual CSS code $\mathcal{Q}^\perp$ of a CSS code $\mathcal{Q}$ is defined by exchanging the roles of the check matrices $H_x$ and $H_z$, thus a CSS code with $H_x = H_z$ is self-dual; see Ref. \cite{BB24} for example. This should not be confused with a CSS code constructed from a classical code $\mathcal{D}$ which is self-dual ($\mathcal{D}^\perp = \mathcal{D}$).}. 

Let $[k]$ denote a set of indices $\{1,\dots,k\}$ where $k \geq 1$. Consider a stabilizer code $\mathcal{Q}$ with stabilizer group $S$. The centralizer $C(S)$ of $S$ is the group of all Pauli operators that commute with all stabilizers. Let $[A,B]=AB-BA$ and $\{A,B\}=AB+BA$ denote commutation and anticommutation of operators $A$ and $B$, respectively. $C(S)$ can be generated by $n-k$ stabilizer generators, $iI^{\otimes n}$, and $k$ pairs of (physical) Pauli operators $(P_j,Q_j)$, $j\in[k]$ such that $\{P_j,Q_j\}=0$ for any $j\in[k]$, $[P_j,Q_{j'}]=0$ for any $j,j'\in[k]$ such that $j\neq j'$, and $[P_j,P_{j'}]=[Q_j,Q_{j'}]=0$ for any $j,j'\in[k]$. We can associate each pair of $(P_j,Q_j)$ with logical Pauli operators $(\bar{X}_j,\bar{Z}_j)$. That is, the group $C(S)/S$ is isomorphic to the Pauli group $\mathcal{P}_k$ defined on $k$ qubits. In this work, we refer to a pair of logical Pauli operators $(\bar{X}_j,\bar{Z}_j)$ as \emph{a hyperbolic pair}, and refer to a set of $k$ hyperbolic pairs $\{(\bar{X}_j,\bar{Z}_j)\}_{j\in[k]}$ as \emph{a symplectic basis}, following the notations from Refs. \cite{BDH06,Wilde09b}. In general, each $\bar{X}_j$ ($\bar{Z}_j$) does not have to be an $X$-type (a $Z$-type) physical Pauli operator. Note that there are many possible choices of a symplectic basis that, together with stabilizer generators and $iI^{\otimes n}$, generates the same centralizer $C(S)$. How logical states of the code (e.g., $\ket{\bar{m}_1,\dots,\bar{m}_k}$ where $m_j \in \{0,1\}$) are defined depends on the choice of a symplectic basis. Any unitary operator $U$ on $n$ physical qubits that preserves the stabilizer group under conjugation, i.e., $UMU^\dagger \in S$ for all $M \in S$, is a \emph{logical operator} of the stabilizer code. The logical operation corresponding to $U$ can be found by considering how each logical Pauli operator is transformed under conjugation. In other words, the logical operation of $U$ depends on the choice of a symplectic basis. 

A \emph{gate gadget} of a unitary operator $L$ is a sequence of quantum operations such as physical gates, qubit preparation, and qubit measurement that implements a logical $L$ operation on the code space. Both input and output of a gate gadget that implements an $m$-qubit gate are $m$ blocks of code. In a non-ideal situation, each component of the gate gadget can be faulty and may lead to high-weight errors which are not correctable by the stabilizer code. To ensure that error propagation is still under control, a gate gadget must be fault tolerant according to the following definition.  

\begin{definition}{Fault-tolerant gate gadget \cite{AGP06}} \label{def:FT_gate_gadget}
	
	Let $t \leq \lfloor (d-1)/2\rfloor$ where $d$ is the distance of a stabilizer code. Suppose that the $i$-th input of a gate gadget implementing an $m$-qubit gate has an error of weight $r_i$ and the gate gadget has $s$ faults. The gate gadget is \emph{$t$-fault tolerant} if the following two conditions are satisfied:
	\begin{enumerate}
		\item Gate Correctness Property (ECCP): For any $r_i, s$ with $\sum_{i=1}^{m}r_i+s \leq t$, ideally decoding the state on each output of the gate gadget gives the same codeword as ideally decoding the state on each input of the gate gadget then applying the corresponding ideal $m$-qubit gate.
        \item Gate Error Propagation Property (GPP): For any $r_i, s$ with $\sum_{i=1}^{m}r_i+s \leq t$, the state on each output of the gate gadget differs from any valid codeword by an error of weight at most $\sum_{i=1}^{m}r_i+s$.
	\end{enumerate}
\end{definition}

On some stabilizer codes, some logical gates can be implemented transversally. The formal definition of transversal gates at the physical level is as follows. 

\begin{definition} \label{def:transversal_at_physical}
    Let $\mathcal{Q}$ be an \codepar{n,k,d} stabilizer code and let $U$ be a quantum gate acting on a single code block of $\mathcal{Q}$. $U$ is \emph{transversal at the physical level} if there exists a decomposition $U=\bigotimes_{i=1}^{n}G_i$ where $G_i$ is a single-qubit gate acting on the $i$-th physical qubit of the code block. Let $V$ be a quantum gate acting on two code blocks of $\mathcal{Q}$. $V$ is \emph{transversal at the physical level} if there exists a decomposition $V=\bigotimes_{i=1}^{n}F_{1:i,2:i}$ where $F_{1:i,2:i}$ is a two-qubit gate acting on the $i$-th physical qubit of the first code block and the $i$-th physical qubit of the second code block.
\end{definition}
(Although the definition presented above can be generalized to a quantum gate on $m$ code blocks composing of $n$ $m$-qubit gates, this work focuses on $m=1,2$ only since we can assume that any unitary operation can be decomposed into a sequence of one- and two-qubit gates.)

A transversal gate is naturally fault tolerant, as any error on a single physical qubit cannot propagate to any other physical qubits in the same code block. Note that although the weight of any Pauli error remains the same after propagation, the type of the error can be changed. For example, a transversal phase gate composing of $S$ and $S^\dagger$ transforms an $X$ error into a $Y$ or a $-Y$ error. Nevertheless, if a CSS code can correct the error before propagation, the error after propagation can also be corrected in the standard circuit-level depolarizing noise model; A single $Y$ error can be viewed as a combination of one $X$ and one $Z$ errors, so a standard QEC procedure of a CSS code which corrects $X$-type and $Z$-type errors separately will treat each $Y$ error as a single error for both $X$-type and $Z$-type decoders. Therefore, the error correction capability of the CSS code after the application of a transversal gate remains the same.

Clifford operations are unitary operations that map any Pauli operator to another Pauli operator. Let $\mathcal{P}_n$ be the Pauli group and $\mathcal{U}_n$ be the unitary group defined on $n$ qubits. The Clifford group $\mathcal{C}_n$ defined on $n$ qubits is,
\begin{equation}
    \mathcal{C}_n = \{A \in \mathcal{U}_n \;|\; AMA^{\dagger} \in \mathcal{P}_n \;\forall M \in \mathcal{P}_n\}.
\end{equation}
$\mathcal{C}_n$ can be generated by Hadamard ($H$) and phase ($S$) gates on any single qubit, and controlled-NOT (CNOT) gates on any pair of qubits.

In this work, we are interested in logical Clifford gates which are transversal at the logical level, as defined below.
\begin{definition} \label{def:types_of_transversal}
Let $\mathcal{Q}$ be an \codepar{n,k,d} stabilizer code with $k \geq 1$ and let $B \subseteq [k]$ be a set of indices of logical qubits.
\begin{enumerate}
    \item A \emph{Pauli-type logical-level transversal gate} $\bar{U}_P\left(B\right)$ is a quantum gate defined on one code block of $\mathcal{Q}$ such that a logical Pauli-type gate ($\bar{X}_j$, $\bar{Y}_j$, or $\bar{Z}_j$) is applied to each logical qubit indexed by $j \in B$, and the identity gates are applied to other logical qubits.
    \item A \emph{Hadamard-type logical-level transversal gate} $\bar{U}_H\left(B\right)$ is a quantum gate defined on one code block of $\mathcal{Q}$ such that a logical Hadamard gate $\bar{H}_j$ is applied to each logical qubit indexed by $j \in B$, and the identity gates are applied to other logical qubits.
    \item A \emph{phase-type logical-level transversal gate} $\bar{U}_S\left(B;\mathbf{a}\right)$ is a quantum gate defined on one code block of $\mathcal{Q}$ such that a logical phase-type gate ($\bar{S}_j$ or $\bar{S}_j^\dagger=\bar{S}_j^{-1}$, specified by $\mathbf{a} \in \{1,-1\}^{|B|}$) is applied to each logical qubit indexed by $j \in B$, and the identity gates are applied to other logical qubits.
    \item A \emph{CNOT-type logical-level transversal gate} $\bar{U}_\mathrm{CNOT}\left(B\right)$ is a quantum gate defined on two code blocks of $\mathcal{Q}$ such that a logical CNOT gate $\overline{\mathrm{CNOT}}_{1:j,2:j}$ is applied to each pair of logical qubits, a control qubit from the first block and a target qubit from the second block, both indexed by $j \in B$, and the identity gates are applied to other logical qubits.
\end{enumerate}
\end{definition}
In particular, we would like to find a transversal implementation of physical Clifford gates that induce similar transversal gates on the logical level. Here we limit our focus to logical-level transversal gates $\bar{U}_P\left(B\right)$ for any $B \subseteq [k]$, $\bar{U}_H\left([k]\right)=\bigotimes_{j=1}^{k} \bar{H}_{j}$, any $\bar{U}_S\left([k];\mathbf{a}\right)$ of the form $\bigotimes_{j=1}^{k} \bar{S}_j^{a_j}$ where $(a_1,\dots,a_k)\in\{-1,1\}^k$, and $\bar{U}_\mathrm{CNOT}\left([k]\right)=\bigotimes_{j=1}^{k} \overline{\mathrm{CNOT}}_{1:j,2:j}$.

The construction of Pauli-type logical-level transversal gates $\bar{U}_P\left(B\right)$ is simple. For any stabilizer code, each logical Pauli operator of the code can be defined with physical Pauli operators, and any tensor product of logical Pauli operators is a product of physical Pauli operators. Thus, for any $B \subseteq [k]$, $\bar{U}_P\left(B\right)$ is also transversal at the physical level. 

The construction of the CNOT-type logical-level transversal gate $\bar{U}_\mathrm{CNOT}\left([k]\right)$ is also simple. For any CSS code, it is always possible to choose a symplectic basis so that any logical Pauli operator $\bar{X}_j$ ($\bar{Z}_j$) can be defined with physical Pauli-$X$ (Pauli-$Z$) operators only. Let $\bar{P}_{1:j}$ ($\bar{P}_{2:j}$) denote a logical $P$ operator acting on the $j$-th logical qubit of the first (second) code block, and 
suppose that logical Pauli operators on two code blocks are defined by the same symplectic basis. Transversal implementation of physical CNOT gates $\bigotimes_{i=1}^{n} \mathrm{CNOT}_{1:i,2:i}$ between two code blocks (where qubits in the first and the second blocks act as control and target qubits, respectively) preserves the stabilizer group defined on two code blocks, and transforms $\bar{X}_{1:j} \otimes \bar{I}_{2:j}$ to $\bar{X}_{1:j} \otimes \bar{X}_{2:j}$, $\bar{I}_{1:j} \otimes \bar{X}_{2:j}$ to $\bar{I}_{1:j} \otimes \bar{X}_{2:j}$, $\bar{Z}_{1:j} \otimes \bar{I}_{2:j}$ to $\bar{Z}_{1:j} \otimes \bar{I}_{2:j}$, and $\bar{I}_{1:j} \otimes \bar{Z}_{2:j}$ to $\bar{Z}_{1:j} \otimes \bar{Z}_{2:j}$ for all $j \in \{1,\dots,k\}$. This transformation implies that $\bigotimes_{j=1}^{k} \overline{\mathrm{CNOT}}_{1:j,2:j}=\bigotimes_{i=1}^{n} \mathrm{CNOT}_{1:i,2:i}$; i.e., $\bar{U}_\mathrm{CNOT}\left([k]\right)$ is also transversal at the physical level.

One main goal of this work is to find a symplectic basis for a self-dual CSS code such that Hadamard-type and phase-type logical-level transversal gates $\bar{U}_H\left([k]\right)$ and $\bar{U}_S\left([k];\mathbf{a}\right)$ for any $\mathbf{a}\in\{-1,1\}^k$ are also transversal at the physical level. More specifically, we are interested in a symplectic basis such that the following implementations are possible.
\begin{definition} \label{def:compatible_basis}
    Let $\mathcal{Q}$ be an \codepar{n,k,d} self-dual CSS code with $k \geq 1$. A symplectic basis of $\mathcal{Q}$ is \emph{compatible with multilevel transversal Clifford operations} if on a single code block of $\mathcal{Q}$,
    \begin{enumerate}
        \item $\bigotimes_{j=1}^{k} \bar{H}_j = \bigotimes_{i=1}^{n} H_i$; and
        \item for any $(a_1,\dots,a_k)\in\{-1,1\}^k$, there exists $(b_1,\dots,b_n)\in\{-1,1\}^n$ such that $\bigotimes_{j=1}^{k} \bar{S}_j^{a_j} = \bigotimes_{i=1}^{n} S_i^{b_i}$.
    \end{enumerate}
\end{definition}
Throughout this work, we will refer to a symplectic basis of the code $\mathcal{Q}$ that satisfies \cref{def:compatible_basis} as a \emph{compatible symplectic basis} of $\mathcal{Q}$. With the same symplectic basis, $\bar{U}_P\left(B\right)$ for any $B \subseteq [k]$ and $\bar{U}_\mathrm{CNOT}\left([k]\right)$ are also transversal at the physical level. Note that for some self-dual CSS codes, a symplectic basis with such properties does not exist (as we will see later in \cref{sec:suf_nec_cond}).

Suppose that a self-dual CSS code $\mathcal{Q}$ has a compatible symplectic basis and logical Pauli operators are defined accordingly. Then, the logical-level transversal gates $\bar{U}_P\left(B\right)$ for any $B \subseteq [k]$, $\bar{U}_H\left([k]\right)=\bigotimes_{j=1}^{k} \bar{H}_j$, any $\bar{U}_S\left([k];\mathbf{a}\right)$ of the form $\bigotimes_{j=1}^{k} \bar{S}_j^{a_j}$ where $(a_1,\dots,a_k)\in\{-1,1\}^k$, and $\bar{U}_\mathrm{CNOT}\left([k]\right)=\bigotimes_{j=1}^{k} \overline{\mathrm{CNOT}}_{1:j,2:j}$ can be fault-tolerantly implemented by some physical-level transversal gates. Let us consider $m$ blocks of code $\mathcal{Q}$ that consist of $mk$ logical qubits in total. Since $\bar{H}^{\otimes k}$, $\bar{S}^{\otimes k}$, and $\overline{\mathrm{CNOT}}^{\otimes k}$ can be implemented on any block (or any pair of blocks) by some physical-level transversal gates, and since the Clifford group $\mathcal{C}_m$ can be generated by $\{H_p,S_p,\mathrm{CNOT}_{p,q}\}_{p,q\in [m]}$, we find that for any logical Clifford operator $\bar{C}$ in the logical Clifford group $\bar{\mathcal{C}}_m$, a logical Clifford gate of the form $\bar{C}^{\otimes k}$ can be implemented by some physical-level transversal gates. The supporting code blocks of $\bar{C}^{\otimes k}$ are determined by the supporting qubits of the corresponding physical Clifford gate $C \in \mathcal{C}_m$, and for each code block that $\bar{C}^{\otimes k}$ acts on, the same logical Clifford operation is applied to all $k$ logical qubits in that block. However, this does not mean that any logical Clifford gate $\bar{D} \in \bar{\mathcal{C}}_k$ defined on each code block can be implemented transversally. To achieve any logical Clifford gate, additional addressable logical gates are required. The following proposition provides ways to achieve the full logical Clifford group on $mk$ logical qubits.

\begin{proposition} \label{prop:full_Clifford}
Let $\mathcal{Q}$ be an \codepar{n,k,d} stabilizer code and suppose that there are $m$ blocks of $\mathcal{Q}$. Also, let $\bar{G}_{p:j}$ denote a logical single-qubit gate acting on the $j$-th logical qubit of the $p$-th code block, and let $\bar{F}_{p:j,q:l}$ denote a logical two-qubit gate acting on the $j$-th logical qubit of the $p$-th code block and the $l$-th logical qubit of the $q$-th code block. The logical Clifford group $\bar{\mathcal{C}}_{mk}$ defined on the total $mk$ logical qubits can be generated by combining logical gates from the following three sets:
\begin{enumerate}
    \item the set of logical transversal Clifford gates,
    \begin{equation}
        \left\{\bigotimes_{j=1}^{k} \bar{H}_{p,j},\; \bigotimes_{j=1}^{k} \bar{S}_{p,j}^{a_j},\; \bigotimes_{j=1}^{k} \overline{\mathrm{CNOT}}_{p:j,q:j} \right\} \label{eq:trans_from_compatible}
    \end{equation}
    $\forall p,q \in [m],\; \forall (a_1,\dots,a_k)\in\{-1,1\}^k$, and
    \item one of the following sets of logical gates inside a code block,
    \begin{align}
        &\left\{\bar{H}_{p,j},\;\overline{\mathrm{CNOT}}_{p:j,p:l} \right\} \;\forall p \in [m], \;\forall j,l \in [k], \label{eq:non-trans1} \\
        \mathrm{or}\;&\left\{\bar{S}_{p,j},\;\overline{\mathrm{CNOT}}_{p:j,p:l} \right\} \;\forall p \in [m], \;\forall j,l \in [k], \label{eq:non-trans2} \\
        \mathrm{or}\;&\left\{\bar{H}_{p,j},\;\overline{\mathrm{CZ}}_{p:j,p:l} \right\} \;\forall p \in [m],\;\forall j,l \in [k], \label{eq:non-trans3}\\
        \mathrm{or}\;&\left\{\overline{\mathrm{CNOT}}_{p:j,p:l},\;\overline{\mathrm{CZ}}_{p:j,p:l} \right\} \;\forall p \in [m],\;\forall j,l \in [k] \label{eq:non-trans4},
    \end{align}
    and
    \item one of the following sets of logical gates between code blocks,
    \begin{align}
        &\left\{\overline{\mathrm{CNOT}}_{p:1,q:1} \right\} \;\forall p,q \in [m], \label{eq:non-trans_btw1} \\
        \mathrm{or}\;&\left\{\overline{\mathrm{CZ}}_{p:1,q:1} \right\} \;\forall p,q \in [m], \label{eq:non-trans_btw2}
    \end{align}
\end{enumerate}
(where $\overline{\mathrm{CZ}}_{p:j,q:l}$ denotes a logical controlled-$Z$ gate on the corresponding code blocks and logical qubits).
\end{proposition}
The proof of \cref{prop:full_Clifford} is provided in \cref{app:full_Clifford}. 

If the code $\mathcal{Q}$ is a self-dual CSS code in which a compatible symplectic basis exists, the logical gates in \cref{eq:trans_from_compatible} can be fault-tolerantly implemented by some physical-level transversal gates. Addressable logical Clifford gates in \cref{eq:non-trans1,eq:non-trans2,eq:non-trans3,eq:non-trans4,eq:non-trans_btw1,eq:non-trans_btw2} can be implemented by logical gate teleportation, but the process requires $2k$ steps of fault-tolerant joint logical Pauli measurements and a specific ancilla state for each logical Clifford gate. When the code is a CSS code, an alternative teleportation-based FTQC scheme proposed by Brun et al. \cite{BZHJL15} can be applied. The FTQC scheme is reviewed in \cref{app:gate_TP}.

\section{A motivating example: the \codepar{15,7,3} quantum Hamming code} \label{sec:qHamming}

In this section, we consider the \codepar{15,7,3} quantum Hamming code as a motivating example, describe how the choice of a symplectic basis can affect the logical operation of a transversal implementation of physical $H$ (or $S$ and $S^\dagger$) gates at the physical level, and provide a construction of a compatible symplectic basis (as in \cref{def:compatible_basis}) of the code. The construction will be later extended to a general self-dual CSS code in \cref{sec:suf_nec_cond}.

The \codepar{15,7,3} quantum Hamming code is a self-dual CSS code in which $X$-type and $Z$-type stabilizer generators are constructed from the check matrix of the $[15,11,3]$ classical Hamming code. One possible choice of stabilizer generators is
\begin{equation}
    \begin{matrix*}[l]
        g^x_1 &= X_{1} X_{3} X_{5} X_{7} X_{9} X_{11} X_{13} X_{15}, \\
        g^x_2 &= X_{2} X_{3} X_{6} X_{7} X_{10} X_{11} X_{14} X_{15}, \\
        g^x_3 &= X_{4} X_{5} X_{6} X_{7} X_{12} X_{13} X_{14} X_{15}, \\
        g^x_4 &= X_{8} X_{9} X_{10} X_{11} X_{12} X_{13} X_{14} X_{15}, \\ 
        g^z_1 &= Z_{1} Z_{3} Z_{5} Z_{7} Z_{9} Z_{11} Z_{13} Z_{15}, \\
        g^z_2 &= Z_{2} Z_{3} Z_{6} Z_{7} Z_{10} Z_{11} Z_{14} Z_{15}, \\
        g^z_3 &= Z_{4} Z_{5} Z_{6} Z_{7} Z_{12} Z_{13} Z_{14} Z_{15}, \\
        g^z_4 &= Z_{8} Z_{9} Z_{10} Z_{11} Z_{12} Z_{13} Z_{14} Z_{15}. 
        \end{matrix*} \label{eq:3D_stb}
\end{equation}

There are many ways to define logical Pauli operators of the \codepar{15,7,3} quantum Hamming code. One way is to define them from gauge operators of the \codepar{15,1,3} 3D color code \cite{Bombin15}, a subsystem code with one logical qubit and six gauge qubits whose stabilizer group is the same as that of the \codepar{15,7,3} quantum Hamming code. In particular, the \codepar{15,7,3} quantum Hamming code can be viewed as the \codepar{15,1,3} 3D color code in which both logical and gauge qubits are used to encode logical information. Let $(\bar{X}_i,\bar{Z}_i), i\in[6]$ denote gauge operators that define gauge qubits and let $(\bar{X}_7,\bar{Z}_7)$ denote the logical $X$-type and $Z$-type Pauli operators of the \codepar{15,1,3} 3D color code. One conventional choice of such operators is,
\begin{equation}
    \begin{matrix*}[l]
        \bar{X}_1 = X_{3} X_{7} X_{11} X_{15}, & \bar{Z}_1 = Z_{12} Z_{13} Z_{14} Z_{15}, \\
        \bar{X}_2 = X_{12} X_{13} X_{14} X_{15}, & \bar{Z}_2 = Z_{3} Z_{7} Z_{11} Z_{15}, \\
        \bar{X}_3 = X_{5} X_{7} X_{13} X_{15}, & \bar{Z}_3 = Z_{10} Z_{11} Z_{14} Z_{15}, \\
        \bar{X}_4 = X_{10} X_{11} X_{14} X_{15}, & \bar{Z}_4 = Z_{5} Z_{7} Z_{13} Z_{15}, \\
        \bar{X}_5 = X_{9} X_{11} X_{13} X_{15}, & \bar{Z}_5 = Z_{6} Z_{7} Z_{14} Z_{15}, \\
        \bar{X}_6 = X_{6} X_{7} X_{14} X_{15}, & \bar{Z}_6 = Z_{9} Z_{11} Z_{13} Z_{15}, \\
        \bar{X}_7 = X^{\otimes 15}, & \bar{Z}_7 = Z^{\otimes 15}.
    \end{matrix*} \label{eq:3D_gauge}
\end{equation}
Stabilizer generators and gauge operators of the \codepar{15,1,3} 3D color code from \cref{eq:3D_stb,eq:3D_gauge} can be illustrated by volume and face operators as in \cref{fig:3D}. 

\begin{figure*}[!htbp]
	\centering
    \begin{subfigure}[b]{0.3\textwidth}
        \includegraphics[width=\textwidth]{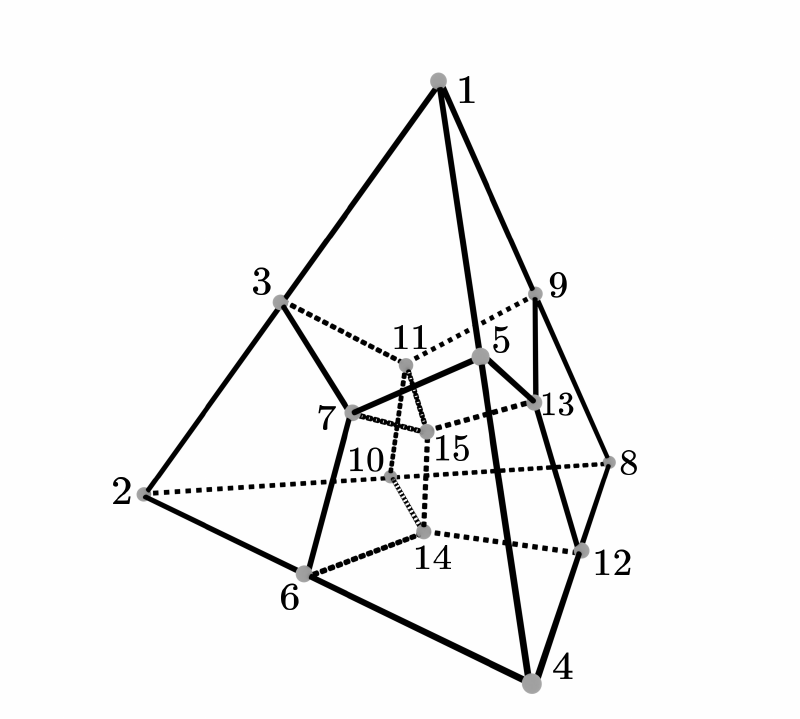}
        \captionsetup{justification=centering}
        \caption{}
    \end{subfigure}
    \begin{subfigure}[b]{0.3\textwidth}
        \includegraphics[width=\textwidth]{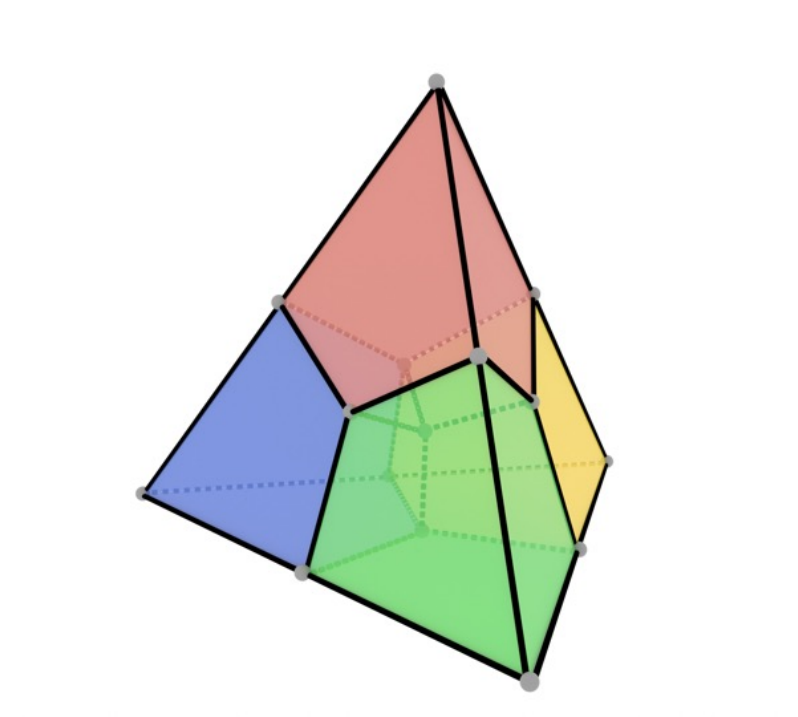}
        \captionsetup{justification=centering}
        \caption{}
    \end{subfigure}
    \begin{subfigure}[b]{0.3\textwidth}
        \includegraphics[width=\textwidth]{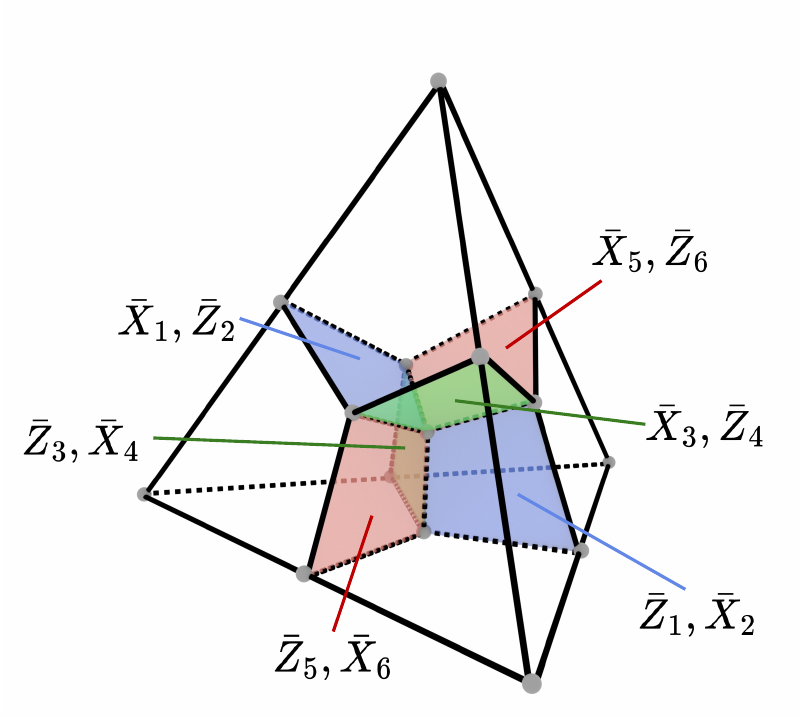}
        \captionsetup{justification=centering}
        \caption{}
    \end{subfigure}
	\caption{The \codepar{15,1,3} 3D color code. (a) Each physical qubit of the code is represented by a vertex with the corresponding numbering. (b) Each stabilizer generator of the code is presented by an 8-body operator. There are $X$-type and $Z$-type stabilizer generators that act on the same supporting qubits. (c) Each gauge generator of the code is presented by a 4-body operator. Each $X$-type gauge generator anticommutes with only one $Z$-type gauge generator (and vice versa). Logical $X$ and logical $Z$ operators are $X^{\otimes 15}$ and $Z^{\otimes 15}$. We can use 1 logical and 6 gauge pairs of the \codepar{15,1,3} 3D color code to define 7 pairs of logical operators of the \codepar{15,7,3} quantum Hamming code. 
    }
	\label{fig:3D}%
\end{figure*}

A symplectic basis $\{(\bar{X}_j,\bar{Z}_j)\}_{j\in[7]}$ of the \codepar{15,7,3} quantum Hamming code can be obtained by defining its logical Pauli operators according to \cref{eq:3D_gauge}. With this definition, we find that (1)
$\{\bar{X}_j,\bar{Z}_j\}=0$ for any $j\in[7]$, $[\bar{X}_j,\bar{Z}_{j'}]=0$ for any $j,j'\in[7]$ such that $j\neq j'$, and $[\bar{X}_j,\bar{X}_{j'}]=[\bar{Z}_j,\bar{Z}_{j'}]=0$ for any $j,j' \in [7]$, and (2) the following pairs of operators share the same supporting qubits: $\bar{X}_1$ and $\bar{Z}_2$, $\bar{X}_2$ and $\bar{Z}_1$, $\bar{X}_3$ and $\bar{Z}_4$, $\bar{X}_4$ and $\bar{Z}_3$, $\bar{X}_5$ and $\bar{Z}_6$, $\bar{X}_6$ and $\bar{Z}_5$, and $\bar{X}_7$ and $\bar{Z}_7$.

Let us consider the action of $\bigotimes_{i=1}^{15} H_i$, a transversal implementation of physical $H$ gates in our choice of symplectic basis. $\bigotimes_{i=1}^{15} H_i$ preserves the stabilizer group under conjugation, thus it is a logical operator. However, each logical Pauli operator in the following pairs is transformed to the other operator in the same pair under the action of $\bigotimes_{i=1}^{15} H_i$: $(\bar{X}_1,\bar{Z}_2)$, $(\bar{X}_2,\bar{Z}_1)$, $(\bar{X}_3,\bar{Z}_4)$, $(\bar{X}_4,\bar{Z}_3)$, $(\bar{X}_5,\bar{Z}_6)$, $(\bar{X}_6,\bar{Z}_5)$, and $(\bar{X}_7,\bar{Z}_7)$. In other words, the logical operation of $\bigotimes_{i=1}^{15} H_i$ in this symplectic basis is $\bigotimes_{j=1}^{7} \bar{H}_j$ followed by swapping logical qubits $(1,2)$, $(3,4)$, and $(5,6)$. Therefore, the symplectic basis defined by \cref{eq:3D_gauge} is not a compatible symplectic basis for the \codepar{15,7,3} quantum Hamming code according to \cref{def:compatible_basis}.

Next, let us consider a new symplectic basis $\{(\bar{X}'_j,\bar{Z}'_j)\}_{j\in[7]}$ defined as follows.
\begin{equation}
    \begin{matrix*}[l]
        \bar{X}'_1 =\bar{X}_1\bar{X}_7, & \bar{Z}'_1 =\bar{Z}_2\bar{Z}_7, \\
        \bar{X}'_2 =\bar{X}_2\bar{X}_7, & \bar{Z}'_2 =\bar{Z}_1\bar{Z}_7, \\
        \bar{X}'_3 = \bar{X}_1\bar{X}_2\bar{X}_3\bar{X}_7, & \bar{Z}'_3 = \bar{Z}_1\bar{Z}_2\bar{Z}_4\bar{Z}_7, \\
        \bar{X}'_4 = \bar{X}_1\bar{X}_2\bar{X}_4\bar{X}_7, & \bar{Z}'_4 = \bar{Z}_1\bar{Z}_2\bar{Z}_3\bar{Z}_7, \\
        \bar{X}'_5 = \bar{X}_1\bar{X}_2\bar{X}_3\bar{X}_4\bar{X}_5\bar{X}_7, & \bar{Z}'_5 = \bar{Z}_1\bar{Z}_2\bar{Z}_3\bar{Z}_4\bar{Z}_6\bar{Z}_7, \\
        \bar{X}'_6 = \bar{X}_1\bar{X}_2\bar{X}_3\bar{X}_4\bar{X}_6\bar{X}_7, & \bar{Z}'_6 = \bar{Z}_1\bar{Z}_2\bar{Z}_3\bar{Z}_4\bar{Z}_5\bar{Z}_7, \\
        \bar{X}'_7 = \bar{X}_1\bar{X}_2\bar{X}_3\bar{X}_4\bar{X}_5\bar{X}_6\bar{X}_7, & \bar{Z}'_7 = \bar{Z}_1\bar{Z}_2\bar{Z}_3\bar{Z}_4\bar{Z}_5\bar{Z}_6\bar{Z}_7,
    \end{matrix*} \label{eq:qHamming_new_basis}
\end{equation}
where $\bar{X}_j$ and $\bar{Z}_j$, $j\in[7]$ are defined by \cref{eq:3D_gauge}. With this new definition, we find that (1) the logical Pauli operators satisfy expected commutation and anticommutation relations, and (2) for all $j$, $\bar{X}'_j$ and $\bar{Z}'_j$ share the same supporting qubits. Therefore, the logical operation of $\bigotimes_{i=1}^{15} H_i$ in the new symplectic basis is $\bigotimes_{j=1}^{7} \bar{H}_j$.

Now, let us consider the action of $\bigotimes_{i=1}^{15} S_i$ in this new symplectic basis. $\bigotimes_{i=1}^{15} S_i$ preserves the stabilizer group as it transforms any $g^z_i$ to itself and transforms any $g^x_i$ to $g^x_ig^z_i$. It also transforms logical Pauli operators as follows:
\begin{equation}
    \begin{matrix*}[l]
        \bar{Z}'_j \mapsto \bar{Z}'_j & \mathrm{for}\; j=1,\dots,7, \\
        \bar{X}'_j \mapsto i\bar{X}'_j\bar{Z}'_j & \mathrm{for}\; j=3,4,7, \\
        \bar{X}'_j \mapsto -i\bar{X}'_j\bar{Z}'_j & \mathrm{for}\; j=1,2,5,6.
    \end{matrix*}
\end{equation}
A logical gate $\bar{S}_j$ ($\bar{S}_j^{\dagger}$) transforms $\bar{X}'_j$ to $i\bar{X}'_j\bar{Z}'_j$ ($-i\bar{X}'_j\bar{Z}'_j$) and transforms $\bar{Z}'_j$ to $\bar{Z}'_j$, thus the logical operation of $\bigotimes_{i=1}^{15} S_i$ in the new symplectic basis is $\bigotimes_{j=1}^{7} \bar{S}_j^{a_j}$ with $a_1=a_2=a_5=a_6=-1$ and $a_3=a_4=a_7=1$. An operator $\bigotimes_{j=1}^{7} \bar{S}_j^{a'_j}$ with any $a'_j$ can be obtained by simply multiplying $\bigotimes_{j=1}^{7} \bar{S}_j^{a_j}$ with logical Pauli operators $\bar{Z}'_j$ whenever $\bigotimes_{j=1}^{7} \bar{S}_j^{a_j}$ does not give the desired phase in the transformation of $\bar{X}'_j$. This is possible because $\bar{S}_j\bar{Z}'_j = \bar{S}_j^{\dagger}$ and $\bar{S}_j^\dagger\bar{Z}'_j = \bar{S}_j$. Note that the resulting operator can still be implemented transversally by physical $S$ and $S^\dagger$ gates since any $\bar{Z}'_j$ is composed of physical $Z$ gates, and $S_i Z_i = S^{\dagger}$ and $S_i^\dagger Z_i = S$. For example, the logical gate $\bigotimes_{j=1}^{7} \bar{S}_j$ can be implemented by $\left(\bigotimes_{i=1}^{15} S_i\right)\left(\bar{Z}'_1\bar{Z}'_2\bar{Z}'_5\bar{Z}'_6\right)=\bigotimes_{i=1}^{15} S_i^{b_i}$ with $b_3=b_6=b_9=b_{12}=-1$ and $b_i=1$ for other $i$'s. Therefore, the symplectic basis defined by \cref{eq:qHamming_new_basis} is a compatible symplectic basis for the \codepar{15,7,3} quantum Hamming code according to \cref{def:compatible_basis}.

\section{Necessary and sufficient conditions for existence of a compatible symplectic basis} \label{sec:suf_nec_cond}

In the previous section, we have shown how a compatible symplectic basis for the \codepar{15,7,3} quantum Hamming code could be constructed. In this section, we will extend the similar techniques to general self-dual CSS codes and derive necessary and sufficient conditions for existence of such a compatible symplectic basis.

We start by presenting the necessary and sufficient conditions for existence of a compatible symplectic basis for any \codepar{n,k,d} self-dual CSS code, which is the first main theorem of this work. In what follows, $\mathcal{P}_n^{(x)}$ and $\mathcal{P}_n^{(z)}$ denote the group of $X$-type and $Z$-type Pauli operators defined on $n$ qubits, $\mathrm{supp}(U)$ denotes the set of supporting (physical) qubits of a unitary operator $U$, and $\mathbf{a}\cdot \mathbf{b} = \sum_{i=1}^n a_i b_i\;(\mathrm{mod}\;2)$ denotes the dot product of two binary vectors $\mathbf{a}, \mathbf{b} \in \mathbb{Z}_2^n$.

\begin{theorem} \label{thm:main1}
    Let $\mathcal{Q}$ be an \codepar{n,k,d} self-dual CSS code with $k \geq 1$ which is constructed from a parity check matrix of a classical binary linear code $\mathcal{D}$ satisfying $\mathcal{D}^\perp \subsetneq \mathcal{D}$. Suppose that $\mathcal{D}^\perp = \langle \mathbf{g}_i \rangle$ and $\mathcal{D} = \langle \mathbf{g}_i,\mathbf{h}_j \rangle$, where $i \in [r]$, $j \in [k]$, $\mathbf{g}_i,\mathbf{h}_j \in \mathbb{Z}_2^n$, and $r=(n-k)/2$.
    The following statements are equivalent:
    \begin{enumerate}
        \item There exists $j \in [k]$ such that $\mathbf{h}_j \cdot \mathbf{h}_j = 1$.
        \item There exists at least one hyperbolic pair $(\bar{L}^x,\bar{L}^z)$ of $\mathcal{Q}$ where $\bar{L}^x \in \mathcal{P}_n^x$, $\bar{L}^z \in \mathcal{P}_n^z$ such that $\mathrm{supp}(\bar{L}^x)=\mathrm{supp}(\bar{L}^z)$.
        \item There exists a symplectic basis $\{(\bar{X}_j,\bar{Z}_j)\}_{j\in[k]}$ of $\mathcal{Q}$ where $\{\bar{X}_j\} \subsetneq \mathcal{P}_n^x$, $\{\bar{Z}_j\} \subsetneq \mathcal{P}_n^z$ such that $\mathrm{supp}(\bar{X}_j)=\mathrm{supp}(\bar{Z}_j)$ for all $j \in [k]$. 
        \item There exists a symplectic basis of $\mathcal{Q}$ which is compatible with multilevel transversal Clifford operations; i.e., a symplectic basis of $\mathcal{Q}$ such that
        \begin{enumerate}
            \item $\bigotimes_{j=1}^{k} \bar{H}_j = \bigotimes_{i=1}^{n} H_i$; and
            \item for any $(a_1,\dots,a_k)\in\{-1,1\}^k$, there exists $(b_1,\dots,b_n)\in\{-1,1\}^n$ such that $\bigotimes_{j=1}^{k} \bar{S}_j^{a_j} = \bigotimes_{i=1}^{n} S_i^{b_i}$.
        \end{enumerate}
    \end{enumerate}
\end{theorem}

To prove \cref{thm:main1}, we require the following five lemmas.

\begin{lemma} \label{lem:alg}
    For any \codepar{n,k,d} self-dual CSS code $\mathcal{Q}$ with $k \geq 1$, there exists a symplectic basis $\{(\bar{X}_j,\bar{Z}_j)\}_{j\in[k]}$ where $\{\bar{X}_j\} \subsetneq \mathcal{P}_n^x$, $\{\bar{Z}_j\} \subsetneq \mathcal{P}_n^z$ such that for each $j$, exactly one of the following is true:
    \begin{enumerate}
        \item $\mathrm{supp}(\bar{X}_j)=\mathrm{supp}(\bar{Z}_j)$; or
        \item there exists exactly one $j'\in[k]$, $j'\neq j$ such that $\mathrm{supp}(\bar{X}_j)=\mathrm{supp}(\bar{Z}_{j'})$ and $\mathrm{supp}(\bar{X}_{j'})=\mathrm{supp}(\bar{Z}_{j})$.
    \end{enumerate}
\end{lemma}
\begin{proof}
    Let $\mathcal{Q}$ be an \codepar{n,k,d} self-dual CSS code constructed from a check matrix $H \in \mathbb{Z}_2^{r\times n}$ of a classical binary linear code $\mathcal{D}$ satisfying $\mathcal{D}^\perp \subseteq \mathcal{D}$ (where $r=(n-k)/2$). As $k\geq 1$, $\mathcal{D}^\perp$ is a proper subset of $\mathcal{D}$ (denoted by $\mathcal{D}^\perp \subsetneq \mathcal{D}$). Let $\mathbf{g}_i$ be rows of $H$. As $H$ is a generator matrix of $\mathcal{D}^\perp$, we can write $\mathcal{D}^\perp = \langle \mathbf{g}_i \rangle$, $i\in[r]$.  As $\mathcal{D}^\perp$ is a subgroup of $\mathcal{D}$, we can write $\mathcal{D} = \langle \mathbf{g}_i,\mathbf{h}_j \rangle$, where $\{\mathbf{h}_j\}$, $j\in[k]$ generates a set of coset representatives of $\mathcal{D}^\perp$ in $\mathcal{D}$. As $\mathbf{g}_i\in \mathcal{D}^\perp \subsetneq \mathcal{D}$ and $\mathbf{h}_j \in \mathcal{D}\setminus\mathcal{D}^\perp$, we have that $\mathbf{g}_i \cdot \mathbf{g}_{i'} = 0$ for all $i,i'\in[r]$ and $\mathbf{g}_i \cdot \mathbf{h}_j = 0$ for all $i\in[r]$, $j\in[k]$. Note that for each pair of $j,j'\in[k]$, the value of $\mathbf{h}_{j} \cdot \mathbf{h}_{j'}$ can be either $0$ or $1$. $\mathbf{g}_i$ and $\mathbf{h}_j$ correspond to stabilizer generators and generators of logical Pauli operators of $\mathcal{Q}$, respectively.

    To construct a symplectic basis of $\mathcal{Q}$ with the desired property, we apply Algorithm 1 from Ref. \cite{TN24}, which is also presented below. Algorithm 1 is a modified version of symplectic Gram-Schmidt orthogonalization, a method to construct a symplectic basis from a candidate set of logical Pauli operators, which is tailored for a self-dual CSS code. A more general version of symplectic Gram-Schmidt orthogonalization for any stabilizer code can be found in Refs. \cite{BDH06,Wilde09b}.
    
\begin{algorithm} \cite{TN24} \label{alg:1}
Let $\{\mathbf{h}_j\}$, $j\in[k]$ be generators of coset representatives of $\mathcal{D}^\perp$ in $\mathcal{D}$.
The algorithm starts by letting $\{\mathbf{w}_1,\dots,\mathbf{w}_{k}\} = \{\mathbf{h}_1,\dots,\mathbf{h}_{k}\}$, $s=k$, and $m=1$. At Round $m$, consider $\{\mathbf{w}_1,\dots,\mathbf{w}_s\}$ and find the smallest $p$ such that $\mathbf{w}_p \cdot \mathbf{w}_p = 1$.
\begin{enumerate}
    \item If the smallest $p$ such that $\mathbf{w}_p \cdot \mathbf{w}_p = 1$ can be found, then do the following.
    \begin{enumerate}
        \item Assign $\Bell_m^x:=\mathbf{w}_p$ and $\Bell_m^z:=\mathbf{w}_p$.
        \item For $q=1,\dots,p-1$, assign $\mathbf{w}'_{q}:=\mathbf{w}_q+(\mathbf{w}_q\cdot \mathbf{w}_p)\mathbf{w}_p$. \\
        For $q=p+1,\dots,s$, assign $\mathbf{w}'_{q-1}:=\mathbf{w}_q+(\mathbf{w}_q\cdot \mathbf{w}_p)\mathbf{w}_p$.\\
        (This is to ensure that $\mathbf{w}'_q \cdot \mathbf{w}_p=0$ for all $q=1,\dots,s-1$.)
        \item Assign $s := s-1$, $m:=m+1$, and $\{\mathbf{w}_1,\dots,\mathbf{w}_{s}\}:= \{\mathbf{w}'_1,\dots,\mathbf{w}'_{s}\}$, then continue to the next round.
    \end{enumerate}
    \item If the smallest $p$ such that $\mathbf{w}_p \cdot \mathbf{w}_p = 1$ cannot be found, then do the following.
    \begin{enumerate}
        \item Find the smallest $p$ such that $\mathbf{w}_1 \cdot \mathbf{w}_p = 1$. 
        \item Assign $\Bell_m^x:=\mathbf{w}_1$, $\Bell_m^z:=\mathbf{w}_p$, $\Bell_{m+1}^x:=\mathbf{w}_p$, and $\Bell_{m+1}^z:=\mathbf{w}_1$.
        \item For $q=2,\dots,p-1$, assign $\mathbf{w}'_{q-1}:=\mathbf{w}_q+(\mathbf{w}_q\cdot \mathbf{w}_p)\mathbf{w}_1$. \\
            For $q=p+1,\dots,s$, assign $\mathbf{w}'_{q-2}:=\mathbf{w}_q+(\mathbf{w}_q\cdot \mathbf{w}_1)\mathbf{w}_p+(\mathbf{w}_q\cdot \mathbf{w}_p)\mathbf{w}_1$.\\
            (This is to ensure that $\mathbf{w}'_q \cdot \mathbf{w}_1=\mathbf{w}'_q \cdot \mathbf{w}_p=0$ for all $q=1,\dots,s-2$.)
        \item Assign $s := s-2$, $m:=m+2$, and $\{\mathbf{w}_1,\dots,\mathbf{w}_{s}\}:= \{\mathbf{w}'_1,\dots,\mathbf{w}'_{s}\}$, then continue to the next round.
    \end{enumerate}
\end{enumerate}
After the loop terminates at $s=0$, output $\{\Bell_j^x\}$ and $\{\Bell_j^z\}$ where $j\in[k]$.
\end{algorithm}
Each logical Pauli operator $\bar{X}_j$ ($\bar{Z}_j$) is constructed from each $\Bell_j^x$ ($\Bell_j^z$) where $0$ and $1$ correspond to $I$ and $X$ ($Z$), resulting in a symplectic basis $\{(\bar{X}_j,\bar{Z}_j)\}_{j\in[k]}$. A brief explanation of \cref{alg:1} is as follows: In each round of iteration, there is no $\mathbf{w}_p$ such that $\mathbf{w}_p\cdot \mathbf{w}_q = 0$ for all $q \in [s]$ (otherwise $\mathbf{w}_p$ is in $\mathcal{D}^\perp$), so either Step 1 or Step 2 of \cref{alg:1} is executed. If Step 1 is executed, $\Bell_m^x$ and $\Bell_m^z$ obtained from the algorithm give a hyperbolic pair ($\bar{X}_m$,$\bar{Z}_m$) such that $\mathrm{supp}(\bar{X}_m)=\mathrm{supp}(\bar{Z}_m)$. $\Bell_m^x \cdot \Bell_m^z=1$ implies that $\{\bar{X}_m$,$\bar{Z}_m\}=0$.
If Step 2 is executed, $\Bell_m^x$, $\Bell_m^z$, $\Bell_{m+1}^x$, and $\Bell_{m+1}^z$ obtained from the algorithm give hyperbolic pairs ($\bar{X}_m$,$\bar{Z}_m$) and ($\bar{X}_{m+1}$,$\bar{Z}_{m+1}$) such that $\mathrm{supp}(\bar{X}_m)=\mathrm{supp}(\bar{Z}_{m+1})$ and $\mathrm{supp}(\bar{X}_{m+1})=\mathrm{supp}(\bar{Z}_{m})$. $\Bell_m^x \cdot \Bell_m^z=1$, $\Bell_{m+1}^x \cdot \Bell_{m+1}^z=1$, $\Bell_m^x \cdot \Bell_{m+1}^z=0$, and $\Bell_{m+1}^x \cdot \Bell_m^z=0$ imply that $\{\bar{X}_m$,$\bar{Z}_m\}=0$, $\{\bar{X}_{m+1}$,$\bar{Z}_{m+1}\}=0$, $[\bar{X}_m$,$\bar{Z}_{m+1}]=0$, and $[\bar{X}_{m+1}$,$\bar{Z}_m]=0$. The modifications of $\mathbf{w}_q$ in both steps ensure that the hyperbolic pairs to be constructed afterwards commute with the existing hyperbolic pairs. As a result, a symplectic basis with the desired properties can be obtained. (For more details on the soundness of \cref{alg:1}, please refer to the proof of Theorem 4 of Ref. \cite{TN24}.) 
\end{proof}

\begin{figure*}[!htbp]
  \centering
  \begin{subfigure}[b]{0.2\textwidth}
    \centering
    \includegraphics[height=3.2cm, keepaspectratio]{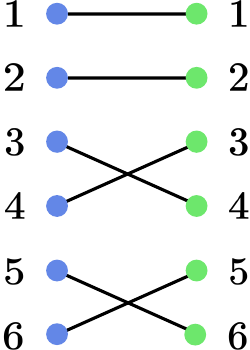}
    \captionsetup{justification=centering}
    \caption{}
  \end{subfigure}
  \begin{subfigure}[b]{0.39\textwidth}
    \centering
    \includegraphics[height=3.2cm, keepaspectratio]{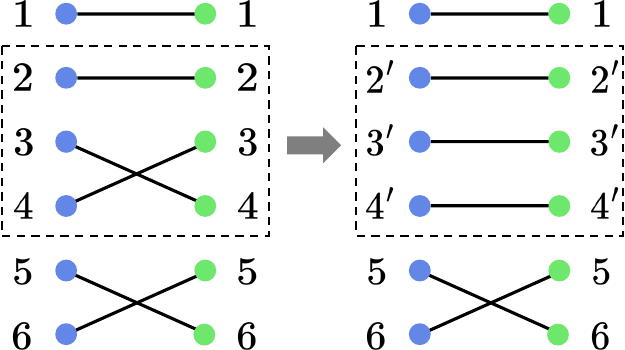}
    \captionsetup{justification=centering}
    \caption{}
  \end{subfigure}
  \begin{subfigure}[b]{0.39\textwidth}
    \centering
    \includegraphics[height=3.2cm, keepaspectratio]{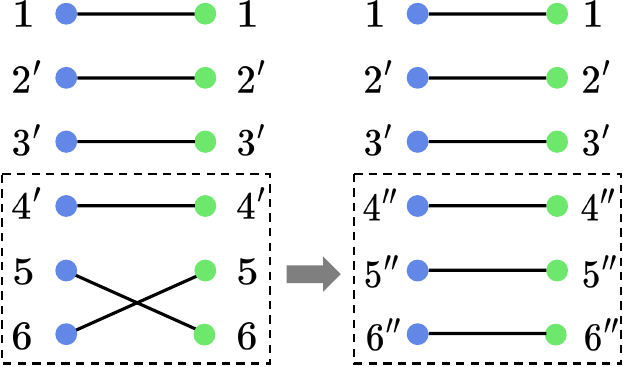}
    \captionsetup{justification=centering}
    \caption{}
  \end{subfigure}
	\caption{(a) For any \codepar{n,k,d} self-dual CSS code, a symplectic basis constructed by \cref{alg:1} has $u$ hyperbolic pairs $(\bar{X}_j,\bar{Z}_j)$ such that $\mathrm{supp}(\bar{X}_j)=\mathrm{supp}(\bar{Z}_j)$, and $v$ pairs of hyperbolic pairs $\left((\bar{X}_{j'},\bar{Z}_{j'}),(\bar{X}_{j''},\bar{Z}_{j''})\right)$ such that $\mathrm{supp}(\bar{X}_{j'})=\mathrm{supp}(\bar{Z}_{j''})$ and $\mathrm{supp}(\bar{X}_{j''})=\mathrm{supp}(\bar{Z}_{j'})$, where $u+2v = k$. The symplectic basis can be illustrated with a bipartite graph in (a), where we consider $u=2$, $v=2$, and $k=6$ in this example. Here, blue (green) vertices represent $\bar{X}_j$ ($\bar{Z}_j$). Logical operators from the same hyperbolic pair are labeled with the same index, and $X$-type and $Z$-type logical operators with exactly the same support are connected by an edge. (b) If $u$ is at least 1, we can apply \cref{lem:basis_change} to the current symplectic basis, resulting in a new symplectic basis represented by a bipartite graph in which $u$ increases by 2 and $v$ decreases by 1. In this example, $u$ becomes 4 and $v$ becomes $1$. (c) By applying \cref{lem:basis_change} repeatedly, a compatible symplectic basis of the self-dual CSS code can be obtained.}
	\label{fig:supp_qubits}%
\end{figure*}

A symplectic basis constructed by \cref{alg:1} consists of $u$ hyperbolic pairs $(\bar{X}_j,\bar{Z}_j)$ such that $\mathrm{supp}(\bar{X}_j)=\mathrm{supp}(\bar{Z}_j)$, and $v$ pairs of hyperbolic pairs $\left((\bar{X}_{j'},\bar{Z}_{j'}),(\bar{X}_{j''},\bar{Z}_{j''})\right)$ such that $\mathrm{supp}(\bar{X}_{j'})=\mathrm{supp}(\bar{Z}_{j''})$ and $\mathrm{supp}(\bar{X}_{j''})=\mathrm{supp}(\bar{Z}_{j'})$, where $u+2v = k$. One can illustrate relations between the supporting qubits of each logical Pauli operator by a bipartite graph in \cref{fig:supp_qubits}.

\begin{lemma} \label{lem:basis_change}
    Let $\mathcal{Q}$ be an \codepar{n,k,d} self-dual CSS code with $k \geq 3$. Suppose that there exist three hyperbolic pairs $(\bar{X}_1,\bar{Z}_1)$, $(\bar{X}_2,\bar{Z}_2)$, $(\bar{X}_3,\bar{Z}_3)$ of $\mathcal{Q}$ where $\{\bar{X}_j\} \subsetneq \mathcal{P}_n^x$, $\{\bar{Z}_j\} \subsetneq \mathcal{P}_n^z$ such that $\mathrm{supp}(\bar{X}_1)=\mathrm{supp}(\bar{Z}_1)$, $\mathrm{supp}(\bar{X}_2)=\mathrm{supp}(\bar{Z}_3)$, and $\mathrm{supp}(\bar{X}_3)=\mathrm{supp}(\bar{Z}_2)$. Then, there exist other three hyperbolic pairs $(\bar{X}'_1,\bar{Z}'_1)$, $(\bar{X}'_2,\bar{Z}'_2)$, $(\bar{X}'_3,\bar{Z}'_3)$ where $\{\bar{X}'_j\} \subsetneq \mathcal{P}_n^x$, $\{\bar{Z}'_j\} \subsetneq \mathcal{P}_n^z$ such that $\langle \bar{X}'_j \rangle = \langle \bar{X}_j \rangle$, $\langle \bar{Z}'_j \rangle = \langle \bar{Z}_j \rangle$, and $\mathrm{supp}(\bar{X}'_j)=\mathrm{supp}(\bar{Z}'_j)$ for all $j=1,2,3$.
\end{lemma}
\begin{proof}
    Suppose that the hyperbolic pairs $(\bar{X}_1,\bar{Z}_1)$, $(\bar{X}_2,\bar{Z}_2)$, and $(\bar{X}_3,\bar{Z}_3)$ of $\mathcal{Q}$ such that $\mathrm{supp}(\bar{X}_1)=\mathrm{supp}(\bar{Z}_1)$, $\mathrm{supp}(\bar{X}_2)=\mathrm{supp}(\bar{Z}_3)$, and $\mathrm{supp}(\bar{X}_3)=\mathrm{supp}(\bar{Z}_2)$ exist. Let $\bar{X}'_1 = \bar{X}_1\bar{X}_2\bar{X}_3$, $\bar{Z}'_1 = \bar{Z}_1\bar{Z}_2\bar{Z}_3$, $\bar{X}'_2 = \bar{X}_1\bar{X}_2$, $\bar{Z}'_2 = \bar{Z}_1\bar{Z}_3$, $\bar{X}'_3 = \bar{X}_1\bar{X}_3$, and $\bar{Z}'_3 = \bar{Z}_1\bar{Z}_2$. We find that $\langle \bar{X}'_j \rangle$ and $\langle \bar{X}_j \rangle$ ($\langle \bar{Z}'_j \rangle$ and $\langle \bar{Z}_j \rangle$) generate the same group, and $(\bar{X}'_1,\bar{Z}'_1)$, $(\bar{X}'_2,\bar{Z}'_2)$, and $(\bar{X}'_3,\bar{Z}'_3)$ satisfy the commutation and anticommutation relations of logical Pauli operators, thus they are hyperbolic pairs of $\mathcal{Q}$. We also have that $\mathrm{supp}(\bar{X}'_j)=\mathrm{supp}(\bar{Z}'_j)$ for all $j=1,2,3$.
\end{proof}

\begin{lemma} \label{lem:stb_preserving_S}
    For any self-dual CSS code $\mathcal{Q}$, there exists $(c_1,\dots,c_n)\in\{-1,1\}^n$ such that $\bigotimes_{i=1}^{n} S^{c_i}$ preserves the stabilizer group $S$ of $\mathcal{Q}$ under conjugation.
\end{lemma}
\begin{proof}
Let $\mathbf{g}_i \in \mathbb{Z}_2^n$ be the binary vector representing both stabilizer generators $g_i^x$ and $g_i^z$ (which is possible since $\mathcal{Q}$ is a self-dual CSS code) where $0$ and $1$ represent $I$ and $X$ or $Z$, $i\in[r]$ where $r=(n-k)/2$ is the number of stabilizer generators of one type, and let $w_i=\mathrm{wt}(\mathbf{g}_i)$ be the Hamming weight of $\mathbf{g}_i$. Since $[g_i^x,g_i^z]=0$ for any $i$, $w_i$ is an even number of the form $4m$ or $4m+2$ for some non-negative integer $m$. Also, let $\mathbf{v} \in \mathbb{Z}_2^n$ be a binary vector such that for all $i$, the number of bits that are 1 for both $\mathbf{v}$ and $\mathbf{g}_i$, denoted by $\beta_i$, is even if $w_i=4m$ and is odd if $w_i=4m+2$ (or equivalently, $\mathbf{v}\cdot\mathbf{g}_i=0$ if $w_i=4m$ and $\mathbf{v}\cdot\mathbf{g}_i=1$ if $w_i=4m+2$). We know that $\mathbf{v}$ exists since $\langle \mathbf{g}_i \rangle$ is a subgroup of $\mathbb{Z}_2^n$ and $\mathbf{v}$ is a coset representative of $\langle \mathbf{g}_i \rangle$ in $\mathbb{Z}_2^n$ (we can think of $\mathbf{v}$ as a binary vector representing a $Z$-type Pauli operator $E_z$ that commutes with $g_i^x$ if $w_i=4m$ and anticommutes with $g_i^x$ if $w_i=4m+2$). Such $\mathbf{v}$ can be found by solving the equation $H\mathbf{v} = \boldsymbol{\delta}$, where $H$ is the check matrix whose rows are $\mathbf{g}_i$'s, and the $i$-th bit of $\boldsymbol{\delta}$ is the desired value of $\mathbf{v}\cdot\mathbf{g}_i$ (either 0 or 1). The equation can always be solved by Gaussian elimination with free variables since $H$ is full row rank (all rows of $H$ are linearly independent), and the number of rows of $H$ is less than the number of columns of $H$. Note that for any self-dual CSS code $\mathcal{Q}$, there are many possible choices of $\mathbf{v}$ that satisfies the equation. These choices of vectors are related by additions of $\mathbf{g}_i$ or $\mathbf{h}_j$ (or both), where $\mathbf{g}_i$ and $\mathbf{h}_j$ are binary vectors that correspond to stabilizer generators and generators of logical Pauli operators of $\mathcal{Q}$, respectively.

Let $W = \bigotimes_{i=1}^{n} S^{c_i}$ where $c_i = (-1)^{v_i}$, $v_i$ is the $i$-th bit of $\mathbf{v}$. For any $g_i^z$, $Wg_i^zW^\dagger=g_i^z$. For any $g_i^x$, $Wg_i^xW^\dagger = (i)^{p_i}g_i^xg_i^z$ for some integer $p_i$. The numbers of $S$ and $S^\dagger$ acting on $g_i^x$ are $w_i-\beta_i$ and $\beta_i$, respectively, so $(i)^{p_i} = (i)^{w_i-\beta_i}(-i)^{\beta_i}$ or $p_i=w_i-2\beta_i$. For $\mathbf{g}_i$ of weight $w_i=4m$ which gives even $\beta_i$, $\beta_i=2s$ for some integer $s$ and thus $p_i=4m-2(2s)=4(m-s)$ is divisible by 4, leading to $(i)^{p_i}=1$. For $\mathbf{g}_i$ of weight $w_i=4m+2$ which gives odd $\beta_i$, $\beta_i=2s+1$ for some integer $s$ and thus $p_i=4m+2-2(2s+1) = 4(m-s)$ is also divisible by 4. Therefore, $W$ preserves the stabilizer group as $Wg_i^xW^\dagger = g_i^xg_i^z$ for any $g_i^x$.
\end{proof}

\begin{figure}[!htbp]
  \centering
  \begin{subfigure}[b]{0.23\textwidth}
    \centering
    \includegraphics[height=3.2cm, keepaspectratio]{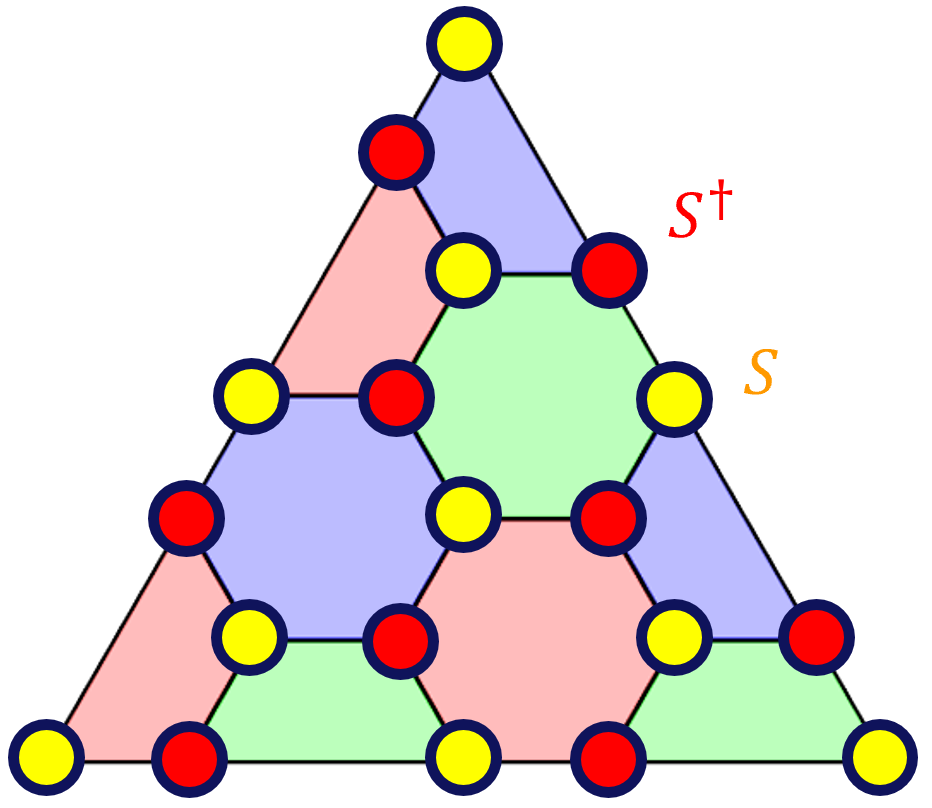}
    \captionsetup{justification=centering}
    \caption{}
  \end{subfigure} 
  \begin{subfigure}[b]{0.23\textwidth}
    \centering
    \includegraphics[height=3.2cm, keepaspectratio]{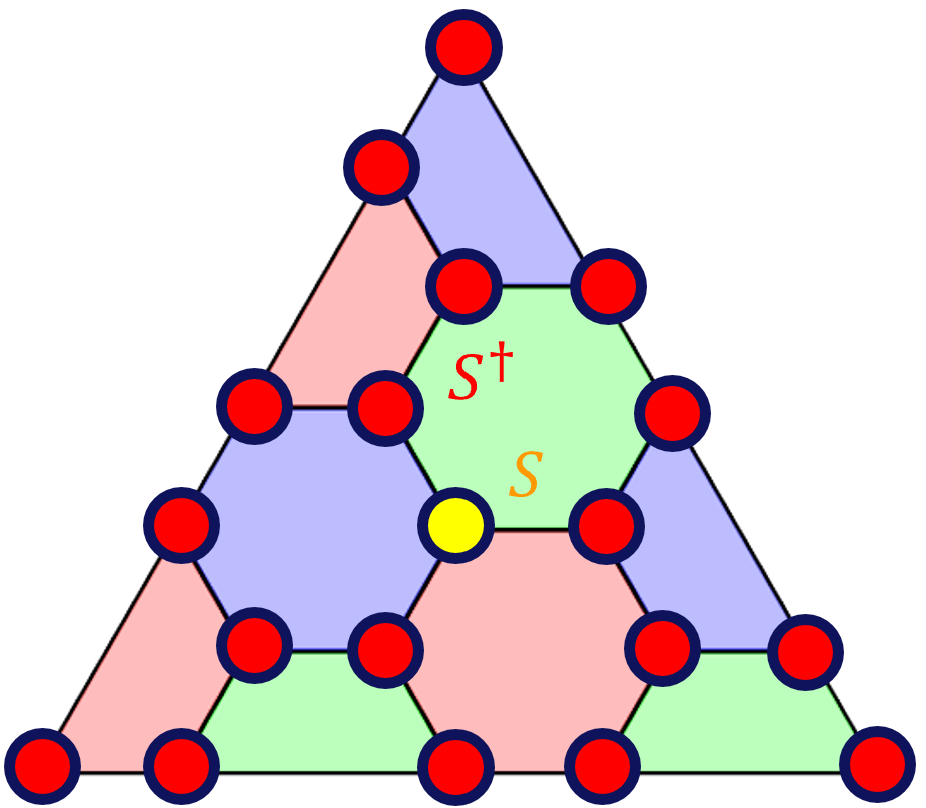}
    \captionsetup{justification=centering}
    \caption{}
  \end{subfigure}
	\caption{Examples of transversal gates of the from $\bigotimes_{i=1}^{n} S^{c_i}$ that preserve the stabilizer group of the \codepar{19,1,5} hexagonal color code. In both examples, $S$ ($S^\dagger$) gates are applied to the qubits on yellow (red) vertices. The vector $\mathbf{v}$ in \cref{lem:stb_preserving_S} is the binary vector whose bit with value one corresponds to a qubit on a red vertex. (a) a transversal gate corresponding to the construction proposed in \cite{KB15} (b) an alternative choice of transversal gate that also preserves the stabilizer group. Note that the logical actions of the gates in (a) and (b) differ by $\bar{Z}$.}
	\label{fig:S_gate_color_code}%
\end{figure}

The construction of $W = \bigotimes_{i=1}^{n} S^{c_i}$ in \cref{lem:stb_preserving_S} also covers two well-known examples of self-dual CSS codes. The first example is when the code is doubly even, i.e., the weight of any stabilizer generator $w_i$ is divisible by 4. In this case, we can pick $\mathbf{v}=\mathbf{0}$ which leads to $W = S^{\otimes n}$. This is consistent with the well-known fact that for any doubly-even self-dual CSS code, $S^{\otimes n}$ preserves the stabilizer group. 

The second example is when the code is a hexagonal color code \cite{BM06}. For this code, one can find a combination of $S$ and $S^\dagger$ gates that preserves the stabilizer group by describing the physical qubits of the code by vertices of a bipartite graph, then applying $S$ gates on qubits in one of the two sets and applying $S^\dagger$ gates on qubits in the other set \cite{KB15}. If we consider a $Z$-type error acting on all qubits that $S^\dagger$ gates are applied to, we will find that for this error, a syndrome bit evaluated by any $X$-type stabilizer generator of weight 4 is even and a syndrome bit evaluated by any $X$-type stabilizer generator of weight 6 is odd. Therefore, a binary vector $\mathbf{v}$ representing this $Z$-type error satisfies $\mathbf{v}\cdot\mathbf{g}_i=0$ if $w_i=4m$ and $\mathbf{v}\cdot\mathbf{g}_i=1$ if $w_i=4m+2$. In this case, the operator $W$ constructed by \cref{lem:stb_preserving_S} from this choice of $\mathbf{v}$ and the operator obtained from the construction in Ref. \cite{KB15} are the same operator. However, there are other choices of $\mathbf{v}$ that can also lead to an operator $W$ that preserves the stabilizer group. In \cref{fig:S_gate_color_code}, we provide two examples of transversal gates of the form $\bigotimes_{i=1}^{n} S^{c_i}$ that preserve the stabilizer group of the \codepar{19,1,5} hexagonal color code.

\begin{lemma} \label{lem:mixed_type}
For any self-dual CSS code $\mathcal{Q}$ with $k \geq 1$, the following statements are equivalent:
    \begin{enumerate}
        \item There exists a hyperbolic pair $(P,Q)$ where $P,Q \in \mathcal{P}_n$ ($P$ and $Q$ are Pauli operators) such that $\bigotimes_{i=1}^{n} H_i$ transforms $P$ to $Q$ and $Q$ to $P$.
        \item There exists a hyperbolic pair $(\bar{L}^x,\bar{L}^z)$ where $\bar{L}^x\in \mathcal{P}_n^x$, $\bar{L}^z\in \mathcal{P}_n^z$ ($\bar{L}^x$ is a Pauli-$X$ operator and $\bar{L}^z$ is a Pauli-$Z$ operator) such that $\bigotimes_{i=1}^{n} H_i$ transforms $\bar{L}^x$ to $\bar{L}^z$ and $\bar{L}^z$ to $\bar{L}^x$.
    \end{enumerate}
\end{lemma}
\begin{proof}
 $(1 \Rightarrow 2)$ Suppose that a hyperbolic pair $(P,Q)$ where $P,Q \in \mathcal{P}_n$ such that $\bigotimes_{i=1}^{n} H_i$ transforms $P$ to $Q$ and $Q$ to $P$ exists. $P$ and $Q$ can be written as $P=A^xB^z$ and $Q=C^xD^z$ (up to some phase) where $A^x,C^x \in \mathcal{P}_n^x$ and $B^z,D^z \in \mathcal{P}_n^z$. As $P$ is a non-trivial logical Pauli operator, $A^x$ or $B^z$ (or both) is a non-trivial logical Pauli operator of a single-type ($X$ or $Z$). Similarly for $Q$, $C^x$ or $D^z$ (or both) is a non-trivial logical Pauli operator of a single type. The assumption that $\bigotimes_{i=1}^{n} H_i$ transforms $P$ to $Q$ and $Q$ to $P$ implies that $\mathrm{supp}(A^x)=\mathrm{supp}(D^z)$ and $\mathrm{supp}(B^z)=\mathrm{supp}(C^x)$. In other words, $\bigotimes_{i=1}^{n} H_i$ transforms $A^x$ to $D^z$ and $B^z$ to $C^x$. 
The anticommutation relation $\{P,Q\}=0$ implies that either $\{A^x,D^z\}=0$ or $\{B^z,C^x\}=0$ (but not both). In case that $\{A^x,D^z\}=0$, both $A^x$ and $D^z$ are non-trivial logical Pauli operators, thus the hyperbolic pair $(\bar{L}^x,\bar{L}^z) = (A^x,D^z)$ has the desired property. Similarly, in case that $\{B^z,C^x\}=0$, the hyperbolic pair $(\bar{L}^x,\bar{L}^z) = (B^z,C^x)$ has the desired property.

$(2 \Leftarrow 1)$ Suppose that a hyperbolic pair $(\bar{L}^x,\bar{L}^z)$ where $\bar{L}^x\in \mathcal{P}_n^x$, $\bar{L}^z\in \mathcal{P}_n^z$ such that $\bigotimes_{i=1}^{n} H_i$ transforms $\bar{L}^x$ to $\bar{L}^z$ and $\bar{L}^z$ to $\bar{L}^x$ exists. As $\mathcal{P}_n^x,\mathcal{P}_n^z \subsetneq \mathcal{P}_n$, we find that the hyperbolic pair $(P,Q) = (\bar{L}^x,\bar{L}^z)$ has the desired property.
\end{proof}

\begin{lemma} \label{lem:vdotv=1}
Let $\mathcal{D}$ be a classical binary linear code satisfying $\mathcal{D}^\perp \subsetneq \mathcal{D}$, and suppose that $\mathcal{D}^\perp = \langle \mathbf{g}_i \rangle$ and $\mathcal{D} = \langle \mathbf{g}_i,\mathbf{h}_j \rangle$, where $i \in [r]$, $j \in [k]$, $\mathbf{g}_i,\mathbf{h}_j \in \mathbb{Z}_2^n$, and $r=(n-k)/2$. There exists $\mathbf{v} \in \mathcal{D}$ such that $\mathbf{v} \cdot \mathbf{v} = 1$ if and only if there exists $j \in [k]$ such that $\mathbf{h}_j \cdot \mathbf{h}_j = 1$
\end{lemma}

\begin{proof}
    Any codeword $\mathbf{v} \in \mathcal{D}$ can be written as $\mathbf{v} = \sum_{i=1}^r a_i\mathbf{g}_i + \sum_{j=1}^k b_j\mathbf{h}_j$ for some coefficients $a_i,b_j \in \mathbb{Z}_2$ (where the sum is bitwise modulo 2). Note that $\mathbf{g}_i$ is in $\mathcal{D}^\perp$ and $\mathbf{g}_i,\mathbf{h}_j$ are in $\mathcal{D}$, so we have that $\mathbf{g}_i \cdot \mathbf{g}_{i'} = \mathbf{g}_i \cdot \mathbf{h}_j = 0$ for any $i,i' \in [r]$, $j \in [k]$. Thus,
\begin{align}
    \mathbf{v} \cdot \mathbf{v} &= \left(\sum_{i=1}^r a_i\mathbf{g}_i + \sum_{j=1}^k b_j\mathbf{h}_j\right) \nonumber \\
    & \;\;\;\; \cdot \left(\sum_{i'=1}^r a_{i'}\mathbf{g}_{i'} + \sum_{j'=1}^k b_{j'}\mathbf{h}_{j'}\right) \nonumber \\
    &= \sum_{j=1}^k b_j\mathbf{h}_j \cdot \sum_{j'=1}^k b_{j'}\mathbf{h}_{j'} \nonumber \\
    &= \left(\sum_{j=1}^k (b_j)^2 (\mathbf{h}_j \cdot \mathbf{h}_j)\right) \nonumber \\
    &\;\;\;\; + \frac{1}{2}\left(\sum_{\substack{j,j'=1\\j \neq j'}}^k b_jb_{j'} (\mathbf{h}_j \cdot \mathbf{h}_{j'}) + b_{j'}b_j (\mathbf{h}_{j'} \cdot \mathbf{h}_j)\right) \nonumber \\
    &= \sum_{j=1}^k (b_j)^2 (\mathbf{h}_j \cdot \mathbf{h}_j). \label{eq:vdotv}
\end{align}
If there exists $\mathbf{v} \in \mathcal{D}$ such that $\mathbf{v} \cdot \mathbf{v} = 1$, by \cref{eq:vdotv}, there must exist $j \in [k]$ such that $\mathbf{h}_j \cdot \mathbf{h}_j = 1$. On the other hand, if there exists $j \in [k]$ such that $\mathbf{h}_j \cdot \mathbf{h}_j = 1$, then $\mathbf{v}=\mathbf{h}_j$ is a codeword in $\mathcal{D}$ that satisfies $\mathbf{v} \cdot \mathbf{v} = 1$.
\end{proof}

Now we are ready to prove the first main theorem of this work.

\begingroup 
\renewcommand\proofname{Proof of \cref{thm:main1}}
\begin{proof}
    $(1 \Rightarrow 2)$ Suppose that there exists $j \in [k]$ such that $\mathbf{h}_j \cdot \mathbf{h}_j = 1$. We can construct a logical Pauli operator $\bar{L}^x \in \mathcal{P}_n^x$ ($\bar{L}^z \in \mathcal{P}_n^z$) from $\mathbf{h}_j$, where $0$ and $1$ correspond to $I$ and $X$ ($Z$). $\bar{L}^x$ and $\bar{L}^z$ have the same support, and since $\mathbf{h}_j \cdot \mathbf{h}_j = 1$, $\bar{L}^x$ and $\bar{L}^z$ anticommute. Thus, $(\bar{L}^x,\bar{L}^z)$ is a hyperbolic pair with the desired property.

    $(2 \Rightarrow 3)$ For any \codepar{n,k,d} self-dual CSS code $\mathcal{Q}$ with $k \geq 1$, we can apply \cref{lem:alg} to construct the first symplectic basis $\{(\bar{X}_j,\bar{Z}_j)\}_{j\in[k]}$. Assume that 
    there exists at least one hyperbolic pair $(\bar{L}^x,\bar{L}^z)$ of $\mathcal{Q}$ where $\bar{L}^x \in \mathcal{P}_n^x$, $\bar{L}^z \in \mathcal{P}_n^z$ such that $\mathrm{supp}(\bar{L}^x)=\mathrm{supp}(\bar{L}^z)$. Without loss of generality, we can choose $\mathbf{h}_1$ in the input $\{\mathbf{h}_1,\dots,\mathbf{h}_{k}\}$ of \cref{alg:1} to be a binary vector representing both $\bar{L}^x$ and $\bar{L}^z$. For any $k \geq 1$, $(\bar{X}_1,\bar{Z}_1)$ from \cref{alg:1} is exactly $(\bar{L}^x,\bar{L}^z)$. In case that $k=1$ or $2$, the symplectic basis obtained from \cref{alg:1} is a symplectic basis such that $\mathrm{supp}(\bar{X}_j)=\mathrm{supp}(\bar{Z}_j)$ for all $j$, so no further process is required. In case that $k \geq 3$, the symplectic basis obtained from \cref{alg:1} may not be a desired symplectic basis, but it can be represented by a bipartite graph in \cref{fig:supp_qubits} with $u \geq 1$ since $\mathrm{supp}(\bar{X}_1)=\mathrm{supp}(\bar{Z}_1)$. We can apply \cref{lem:basis_change} to one hyperbolic pair $(\bar{X}_a,\bar{Z}_a)$ such that $\mathrm{supp}(\bar{X}_a)=\mathrm{supp}(\bar{Z}_a)$ and one pair of hyperbolic pairs $\left((\bar{X}_{b},\bar{Z}_{b}),(\bar{X}_{c},\bar{Z}_{c})\right)$ such that $\mathrm{supp}(\bar{X}_b)=\mathrm{supp}(\bar{Z}_c)$ and $\mathrm{supp}(\bar{X}_c)=\mathrm{supp}(\bar{Z}_b)$ to construct three hyperbolic pairs $\left((\bar{X}''_{a},\bar{Z}''_{a}),(\bar{X}''_{b},\bar{Z}''_{b}),(\bar{X}''_{c},\bar{Z}''_{c})\right)$ such that $\mathrm{supp}(\bar{X}''_j)=\mathrm{supp}(\bar{Z}''_j)$ for all $j\in\{a,b,c\}\subsetneq [k]$. Each application of \cref{lem:basis_change} leads to a new bipartite graph in which $u$ increases by 2 and $v$ decreases by 1. By applying \cref{lem:basis_change} repeatedly, a symplectic basis $\{(\bar{X}'_j,\bar{Z}'_j)\}_{j\in[k]}$ where $\{\bar{X}'_j\} \subsetneq \mathcal{P}_n^x$, $\{\bar{Z}'_j\} \subsetneq \mathcal{P}_n^z$ such that $\mathrm{supp}(\bar{X}'_j)=\mathrm{supp}(\bar{Z}'_j)$ for all $j \in [k]$ can be obtained.
    
    $(3 \Rightarrow 4)$
    Assume that there exists a symplectic basis $\{(\bar{X}_j,\bar{Z}_j)\}_{j\in[k]}$ of $\mathcal{Q}$ where $\{\bar{X}_j\} \subsetneq \mathcal{P}_n^x$, $\{\bar{Z}_j\} \subsetneq \mathcal{P}_n^z$ such that $\mathrm{supp}(\bar{X}_j)=\mathrm{supp}(\bar{Z}_j)$ for all $j \in [k]$. $\bigotimes_{i=1}^{n} H_i$ preserves the stabilizer group and transforms $\bar{X}_j$ to $\bar{Z}_j$ ($\bar{Z}_j$ to $\bar{X}_j$) for any $j$. Thus, $\bigotimes_{i=1}^{n} H_i=\bigotimes_{j=1}^{k} \bar{H}_j$. Next, let us consider a phase-type gate. From \cref{lem:stb_preserving_S}, there exists an operator $W = \bigotimes_{i=1}^{n} S^{c_i}$ with $(c_1,\dots,c_n)\in\{-1,1\}^n$ that preserves the stabilizer group. Let $c_i = (-1)^{v_i}$ and $\mathbf{v}$ be defined as in the proof of \cref{lem:stb_preserving_S}. Also, for any hyperbolic pair $(\bar{X}_j,\bar{Z}_j)$, let $\Bell_j \in \mathbb{Z}_2^n$ be a binary vector representing both $\bar{X}_j$ and $\bar{Z}_j$ (which is possible since $\mathrm{supp}(\bar{X}_j)=\mathrm{supp}(\bar{Z}_j)$). As $\{\bar{X}_j,\bar{Z}_j\}=1$, the Hamming weight $\mathrm{wt}(\Bell_j)$ of $\Bell_j$ is always odd.

    For any $\bar{Z}_j$, $W\bar{Z}_jW^\dagger = \bar{Z}_j$. For any $\bar{X}_j$, $W\bar{X}_jW^\dagger = (i)^{q_j}\bar{X}_j\bar{Z}_j$ for some integer $q_j$. Let $v_i$ and $\ell_{ji}$ denote the $i$-th bits of $\mathbf{v}$ and $\Bell_j$. The numbers of $S$ and $S^\dagger$ acting on $\bar{X}_j$ are $\mathrm{wt}(\Bell_j)-\sum_{i=1}^{n} v_i \ell_{ji}$ and $\sum_{i=1}^{n} v_i \ell_{ji}$. Thus, $q_j = \mathrm{wt}(\Bell_j)-2\sum_{i=1}^{n} v_i \ell_{ji}$, which is an odd number of the form $4m+1$ or $4m+3$ for some integer $m$. Let $\tilde{q}_j = [(q_j+2)\;\mathrm{mod}\;4]-2$ (so that $\tilde{q}_j = 1$ if $q_j=4m+1$ and $\tilde{q}_j = -1$ if $q_j=4m+3$). The logical operation of $W$ is thus $\bigotimes_{j=1}^{k} \bar{S}_j^{\tilde{q}_j}$. For any $(a_1,\dots,a_k)\in\{-1,1\}^k$, $\bigotimes_{j=1}^{k} \bar{S}_j^{a_j}$ can be written as $\bigotimes_{j=1}^{k} \bar{S}_j^{\tilde{q}_j+2r_j}$ for some $r_j \in \{0,1\}$. This logical operation is the same as $W\left(\prod_{j=1}^k\bar{Z}_j^{r_j}\right)$, which can be implemented transversally by physical $S$ and $S^\dagger$ gates since any $\bar{Z}_j$ is composed of physical $Z$ gates, and $S_i Z_i = S_i^{\dagger}$ and $S_i^\dagger Z_i = S_i$. Therefore, the symplectic basis $\{(\bar{X}_j,\bar{Z}_j)\}_{j\in[k]}$ is a compatible symplectic basis according to \cref{def:compatible_basis}. 
    
    $(4 \Rightarrow 1)$ By contrapositive, assume that there is no $j \in [k]$ such that $\mathbf{h}_j \cdot \mathbf{h}_j = 1$. By \cref{lem:vdotv=1}, there is no $\mathbf{v} \in \mathcal{D}$ such that $\mathbf{v} \cdot \mathbf{v} = 1$. Consequently, a hyperbolic pair $(\bar{L}^x,\bar{L}^z)$ of $\mathcal{Q}$ where $\bar{L}^x \in \mathcal{P}_n^x$, $\bar{L}^z \in \mathcal{P}_n^z$ such that $\mathrm{supp}(\bar{L}^x)=\mathrm{supp}(\bar{L}^z)$ does not exist. This implies that a hyperbolic pair $(\bar{M}^x,\bar{M}^z)$ of $\mathcal{Q}$ where $\bar{M}^x \in \mathcal{P}_n^x$, $\bar{M}^z \in \mathcal{P}_n^z$ such that $\bigotimes_{j=1}^{k} H_i$ transforms $\bar{M}^x$ to $\bar{M}^z$ and $\bar{M}^z$ to $\bar{M}^x$ does not exist. By \cref{lem:mixed_type}, we find that a hyperbolic pair $(P,Q)$ where $P,Q \in \mathcal{P}_n$ such that $\bigotimes_{i=1}^{n} H_i$ transforms $P$ to $Q$ and $Q$ to $P$ does not exist. Therefore, a compatible symplectic basis for $\mathcal{Q}$ does not exist. 
\end{proof}
\endgroup

By \cref{thm:main1}, one can verify the existence of a compatible symplectic basis for any self-dual CSS code by simply showing that the code has at least one anticommuting pair of logical $X$ and logical $Z$ operators that have the same support. Furthermore, a compatible symplectic basis of the code can be constructed by applying \cref{lem:alg,lem:basis_change}. A full procedure for constructing a compatible symplectic basis and phase-type logical-level transversal gates for a self-dual CSS code (if they exist) is provided in \cref{app:python}. A Python implementation of the procedure is also available at \url{https://github.com/yugotakada/mlvtrans}.

\cref{thm:main1} also leads to the following corollary.

\begin{corollary} \label{cor:odd_n}
    Let $\mathcal{Q}$ be an \codepar{n,k,d} self-dual CSS code with $k \geq 1$. If $n$ is odd, then there exists a symplectic basis of $\mathcal{Q}$ which is compatible with multilevel transversal Clifford operations.
\end{corollary}
\begin{proof}
    For any \codepar{n,k,d} self-dual CSS code with odd $n$, any stabilizer generator must have even weight. This is because an $X$-type stabilizer generator and a $Z$-type stabilizer generator with the same support must commute. Thus, we have that $(\bar{L}^x,\bar{L}^z)=(X^{\otimes n},Z^{\otimes n})$ is a hyperbolic pair satisfying $\mathrm{supp}(\bar{L}^x)=\mathrm{supp}(\bar{L}^z)$. By \cref{thm:main1}, a compatible symplectic basis of the code exists.
\end{proof}

We point out that while having odd $n$ is sufficient for the existence of a compatible symplectic basis, this is not a necessary condition. Below, we show that the \codepar{4,2,2} code \cite{LNCY97} is an example of a self-dual CSS code with even $n$ such that a compatible symplectic basis does not exist, and the \codepar{6,2,2} code \cite{Knill05b} is an example of a self-dual CSS code with even $n$ such that a compatible symplectic basis exists.

The \codepar{4,2,2} code \cite{LNCY97} can be described by stabilizer generators,
\begin{equation}
    \begin{matrix*}[l]
        g^x_1 = X_{1} X_{2} X_{3} X_{4}, & g^z_1 = Z_{1} Z_{2} Z_{3} Z_{4}. 
    \end{matrix*}\label{eq:422_stb}
\end{equation}
To show that a compatible symplectic basis for this code does not exist, we will apply \cref{thm:main1} and show that there is no hyperbolic pair $(\bar{L}^x,\bar{L}^z)$ where $\bar{L}^x \in \mathcal{P}_n^x$, $\bar{L}^z \in \mathcal{P}_n^z$ such that $\mathrm{supp}(\bar{L}^x)=\mathrm{supp}(\bar{L}^z)$. Let us consider $\bar{L}^x$ first. For $\bar{L}^x$ to be an $X$-type logical Pauli operator, it must commute with $g^z_1$ but cannot be $g^x_1$. Thus, $\bar{L}^x$ must be an $X$-type operator of weight 2. Similarly, we find that $\bar{L}^z$ must be a $Z$-type operator of weight 2. On the other hand, suppose that $\bar{L}^x$ and $\bar{L}^z$ have the same support. Because they must anticommute, we have that both $\bar{L}^x$ and $\bar{L}^z$ must have odd weight, causing contradiction. Therefore, a compatible symplectic basis for the \codepar{4,2,2} code does not exist.

Next, we consider the \codepar{6,2,2} code and show that it has a compatible symplectic basis. The \codepar{6,2,2} code \cite{Knill05b} can be described by stabilizer generators,
\begin{equation}
    \begin{matrix*}[l]
        g^x_1 = X_{1} X_{2} X_{5} X_{6} , & g^z_1 = Z_{1} Z_{2} Z_{5} Z_{6}, \\
        g^x_2 = X_{3} X_{4} X_{5} X_{6}, & g^z_2 = Z_{3} Z_{4} Z_{5} Z_{6}.
    \end{matrix*}\label{eq:622_stb}
\end{equation}
One possible compatible symplectic basis of this code is $\{(\bar{X}_j,\bar{Z}_j)\}_{j\in[2]}$, where
\begin{equation}
    \begin{matrix*}[l]
        \bar{X}_1 = X_{1} X_{3} X_{5}, & \bar{Z}_1 = Z_{1} Z_{3} Z_{5}, \\
        \bar{X}_2 = X_{2} X_{4} X_{6}, & \bar{Z}_2 = Z_{2} Z_{4} Z_{6}.
        \label{eq:622_log}
    \end{matrix*}
\end{equation}

\section{Clifford gates with multilevel transversality} \label{sec:multilevel_trans}

In previous sections, we have considered a self-dual CSS code and try to find a symplectic basis in which transversal application of some Clifford gates at the physical level lead to similar logical Clifford operations which are transversal at the logical level; i.e., on a compatible symplectic basis, such Clifford operations are transversal at two levels. In this section, we extend the idea to a concatenated code constructed from self-dual CSS codes. We first define transversal gates at each level of concatenation, then show that it is possible to obtain Clifford gates which are transversal at more than two levels if every code in the concatenation has a compatible symplectic basis. Applications of Clifford gates with multilevel transversality will be further discussed in the next section. 

We start by defining a concatenated code with multiple levels of concatenation, as well as logical qubits at each level of concatenation.

\begin{definition} \label{def:concat_code}
Let $\mathcal{Q}_i$ be an \codepar{n_i,k_i,d_i} stabilizer code where $i \in [L]$. An \emph{L-level concatenated code} $\mathcal{Q}_\mathrm{con}^{(L)} = \mathcal{Q}_L \circ \cdots \circ \mathcal{Q}_1$ is an \codepar{N,K,D} stabilizer code with $N = \prod_{i=1}^L n_i$, $K = \prod_{i=1}^L k_i$, and $D = \prod_{i=1}^L d_i$.
A \emph{logical qubit at the $l$-th level} is a logical qubit of the inner code $\mathcal{Q}_l \circ \cdots \circ \mathcal{Q}_1$ of the concatenated code $\mathcal{Q}_\mathrm{con}^{(L)}$. The number of logical qubits at the $l$-th level of the concatenated code is defined as $N^{(l)}=\prod_{i=1}^l k_i \prod_{j=l+1}^L n_j$. 
\end{definition}

By \cref{def:concat_code}, the physical qubits of $\mathcal{Q}_\mathrm{con}^{(L)}$ are logical qubits at the zeroth level, and the actual logical qubits of $\mathcal{Q}_\mathrm{con}^{(L)}$ are logical qubits at the $L$-th level. 

$\mathcal{Q}_\mathrm{con}^{(L)}$ encodes $N^{(0)}=N$ physical qubits, each labeled by $(i_1,\dots,i_L) \in [n_1]\times\cdots\times[n_L]$, to $N^{(L)}=K$ logical qubits, each labeled by $(j_1,\dots,j_L) \in [k_1]\times\cdots\times[k_L]$. In particular, for each $l \in \{0,\dots,L\}$ and for any $j_1\in[k_1],\dots,j_{l-1}\in[k_{l-1}],i_{l+1}\in[n_{l+1}],\dots,i_L\in[n_L]$, $\mathcal{Q}_l$ encodes $n_l$ logical qubits at the $(l-1)$-th level with indices $(j_1,\dots,j_{l-1},i_l,i_{l+1},\dots,i_L)$, $i_l \in [n_l]$ to $k_l$ logical qubits at the $l$-th level with indices $(j_1,\dots,j_{l-1},j_l,i_{l+1},\dots,i_L)$, $j_l \in [k_l]$. An example of the encoding of logical qubits of a $2$-level concatenated codes is given in \cref{fig:code_concat}. In this example, $\mathcal{Q}_1$ is the \codepar{15,7,3} Hamming code and $\mathcal{Q}_2$ is the \codepar{6,2,2} code. The descriptions are as follows:
\begin{enumerate}
    \item The number of logical qubits at the zeroth level is $N^{(0)}=(15)(6)=90$, which is the number of physical qubits. In this level, each qubit is represented by a blue circle, and is labeled by $(i_1,i_2)$, where $i_1$ ($i_2$) refers the column (row) of the left grid of qubits.
    \item The number of logical qubits at the first level is $N^{(1)}=(7)(6)=42$. In this level, each qubit is represented by a green diamond, and is labeled by $(j_1,i_2)$, where $j_1$ ($i_2$) refers the column (row) of the middle grid of qubits.
    \item The number of logical qubits at the second level is $N^{(2)}=(7)(2)=14$. In this level, each qubit is represented by a red triangle, and is labeled by $(j_1,j_2)$, where $j_1$ ($j_2$) refers the column (row) of the right grid of qubits.
\end{enumerate}

\begin{figure*}[!htbp]
	\centering
	\includegraphics[width=0.75\textwidth]{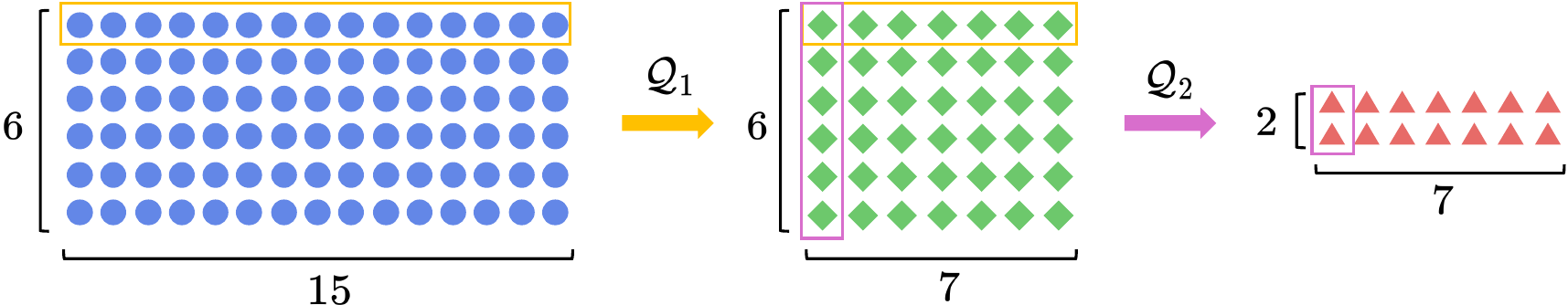}
	\caption{The encoding of logical qubits of a \codepar{90,14,\geq 6} concatenated code $\mathcal{Q}_2 \circ \mathcal{Q}_1$, where $\mathcal{Q}_1$ is the \codepar{15,7,3} Hamming code and $\mathcal{Q}_2$ is the \codepar{6,2,2} code. The logical qubits at the zeroth, the first, and the second level of concatenation are represented by blue circles, green diamonds, and red triangles, respectively. $\mathcal{Q}_1$ encodes each row of 15 zeroth-level logical qubits (i.e., physical qubits) on the left grid to a row of 7 first-level logical qubits on the middle grid, and $\mathcal{Q}_2$ encodes each column of 6 first-level logical qubits on the middle grid to a column of 2 second-level logical qubits on the right grid.}
	\label{fig:code_concat}%
\end{figure*}

Next, we define transversal gates at each level of concatenation.
\begin{definition} \label{def:transversal_at_each_level}
Let $\mathcal{Q}_\mathrm{con}^{(L)}$ be an \codepar{N,K,D} concatenated code as defined in \cref{def:concat_code}, and
let $U^{(l)}$ be a quantum gate defined on one block of the concatenated code that acts on $N^{(l)}$ logical qubits at the $l$-th level (which is equivalent to $N^{(m)}$ logical qubits at the $m$-th level). $U^{(l)}$ is \emph{transversal at the $m$-th level} if there exists a decomposition $U^{(l)}=\bigotimes_{i=1}^{N^{(m)}}G_i^{(m)}$ where $G_i^{(m)}$ is a single-qubit gate acting on the $i$-th logical qubit at the $m$-th level. Let $V^{(l)}$ be a quantum gate defined on two blocks of the concatenated code that, on each block, acts on $N^{(l)}$ logical qubits at the $l$-th level (which is equivalent to $N^{(m)}$ logical qubits at the $m$-th level). $V^{(l)}$ is \emph{transversal at the $m$-th level} if there exists a decomposition $V^{(l)}=\bigotimes_{i=1}^{N^{(m)}}F^{(m)}_{1:i,2:i}$, where $F^{(m)}_{1:i,2:i}$ is a two-qubit gate acting on the $i$-th logical qubit at the $m$-th level of the first code block, and the $i$-th logical qubit at the $m$-th level of the second code block.
\end{definition}

Transversality of a gate at the $m$-level implies that an implementation of the gate is fault tolerant to some extent, as stated in the proposition below.
\begin{proposition} \label{prop:FT_of_transversal}
    If a quantum gate acting on one code block $U^{(l)}$ is transversal at the $m$-th level, then there exists a gate gadget implementing $U^{(l)}$ such that $s$ faults in the gadget lead to errors on no more than $s$ logical qubits at the $m$-th level on the supporting code block of $U^{(l)}$. If a quantum gate acting on two code blocks $V^{(l)}$ is transversal at the $m$-th level, then there exists a gate gadget implementing $V^{(l)}$ such that $s$ faults in the gadget lead to errors on no more than $s$ logical qubits at the $m$-th level on each of the supporting code blocks of $V^{(l)}$.
\end{proposition}
Also, transversality at the $m$-level is preserved under operator multiplication.
\begin{proposition} \label{prop:product_of_transversal}
    A multiplication of gates which are transversal at the $m$-th level is also transversal at the $m$-th level.
\end{proposition}

If a certain logical gate is transversal at several levels, the best practice to implement it would be using transversal logical gates at the lowest possible level. This is because in case that some faults occur in the implementation, the spread of errors could be limited to the level on which the transversal logical gates are operated. In general, logical errors at a lower level of concatenation can be corrected more easily than logical errors at a higher level since a QEC gadget at a lower level has a simpler construction. Note that in case that the lowest level for transversal gate implementation is not the physical level, an additional technique such as teleportation-based FTQC in Ref. \cite{BZHJL15} may be required for implementing some gates.

In this work, we are interested in the following types of transversal gates.
\begin{definition} \label{def:types_of_transversal_concat}
Let $\mathcal{Q}_\mathrm{con}^{(L)}=\mathcal{Q}_L \circ \cdots \circ \mathcal{Q}_1$ be an $L$-level concatenated code with parameter \codepar{N,K,D} as defined in \cref{def:concat_code}. Let $B_1 \subseteq [k_1],\dots,B_l \subseteq [k_l],A_{l+1} \subseteq [n_{l+1}],\dots,A_L \subseteq [n_L]$ be sets of indices of logical qubits at the $l$-th level, where $l \in \{0,\dots,L\}$.
\begin{enumerate}
    \item A \emph{Pauli-type transversal gate at the $l$-th level} $U_P^{(l)}\left(B_1,\dots,B_l,A_{l+1},\dots,A_L \right)$ is a quantum gate defined on one code block of $\mathcal{Q}_\mathrm{con}^{(L)}$ such that a Pauli-type gate ($X^{(l)}$, $Y^{(l)}$, or $Z^{(l)}$) is applied to each logical qubit at the $l$-th level indexed by $(j_1,\dots,j_l,i_{l+1},\dots,i_L) \in B_1\times \cdots \times B_l \times A_{l+1} \times \cdots \times A_L$, and the identity gates are applied to other logical qubits at the $l$-th level.
    \item A \emph{Hadamard-type transversal gate at the $l$-th level} $U_H^{(l)}\left(B_1,\dots,B_l,A_{l+1},\dots,A_L \right)$ is a quantum gate defined on one code block of $\mathcal{Q}_\mathrm{con}^{(L)}$ such that a Hadamard gate $H^{(l)}$ is applied to each logical qubit at the $l$-th level indexed by $(j_1,\dots,j_l,i_{l+1},\dots,i_L) \in B_1\times \cdots \times B_l \times A_{l+1} \times \cdots \times A_L$, and the identity gates are applied to other logical qubits at the $l$-th level.
    \item A \emph{phase-type transversal gate at the $l$-th level} $U_S^{(l)}\left(B_1,\dots,B_l,A_{l+1},\dots,A_L;\mathbf{a} \right)$ is a quantum gate defined on one code block of $\mathcal{Q}_\mathrm{con}^{(L)}$ such that a phase-type gate ($S^{(l)}$ or $S^{\dagger(l)}$, specified by $\mathbf{a} \in \{1,-1\}^{c}$) is applied to each logical qubit at the $l$-th level indexed by $(j_1,\dots,j_l,i_{l+1},\dots,i_L) \in B_1\times \cdots \times B_l \times A_{l+1} \times \cdots \times A_L$, and the identity gates are applied to other logical qubits at the $l$-th level, where $c = |B_1\times \cdots \times B_l \times A_{l+1} \times \cdots \times A_L|$.
    \item A \emph{CNOT-type transversal gate at the $l$-th level} $U_\mathrm{CNOT}^{(l)}\left(B_1,\dots,B_l,A_{l+1},\dots,A_L \right)$ is a quantum gate defined on two code blocks of $\mathcal{Q}_\mathrm{con}^{(L)}$ such that a CNOT gate $\mathrm{CNOT}^{(l)}$ is applied to each pair of logical qubits at the $l$-th level, with a control qubit from the first block and a target qubit from the second block, both indexed by $(j_1,\dots,j_l,i_{l+1},\dots,i_L) \in B_1\times \cdots \times B_l \times A_{l+1} \times \cdots \times A_L$, and the identity gates are applied to other logical qubits at the $l$-th level.
\end{enumerate}
\end{definition}
For example, consider a concatenated code in \cref{fig:code_concat} and suppose that $B_1 = \{1,2\}$ and $A_2 = \{3,4\}$. $U_H^{(1)}\left(B_1,A_2 \right)$ refers to a tensor product of Hadamard gates acting on the first-level logical qubits labeled by $(1,3), (1,4), (2,3),$ and $(2,4)$.

With the definitions presented above, the following theorems can be obtained.

\begin{theorem} \label{thm:main2}
Let $\mathcal{Q}_\mathrm{con}^{(L)}=\mathcal{Q}_L \circ \cdots \circ \mathcal{Q}_1$ be an $L$-level concatenated code with parameter \codepar{N,K,D} as defined in \cref{def:concat_code}, and suppose that any code $\mathcal{Q}_i$ ($i \in [L]$) is a self-dual CSS code satisfying the condition in \cref{thm:main1} and logical Pauli operators of each code are defined by a compatible symplectic basis. Then, for any $l \in \{0,\dots,L\}$ and for any $B_1\subseteq[k_1],\dots,B_{l}\subseteq[k_l],A_{l+1}\subseteq[n_{l+1}],\dots,A_L\subseteq[n_L]$, $U_P^{(l)}\left(B_1,\dots,B_l,A_{l+1},\dots,A_L\right)$ is transversal at levels $0,\dots,l$.
\end{theorem}
\begin{proof}
    For any $l \in \{0,\dots,L\}$, any $X$-type ($Z$-type) logical Pauli operator of $\mathcal{Q}_l$ can be defined in a way that it is composed of only physical $X$ ($Z$) and identity gates. Thus, a tensor product of $X^{(l)}$ ($Z^{(l)}$) on a single logical qubit at the $l$-th level and identity gates on other logical qubits is composed of $X^{(l-1)}$ ($Z^{(l-1)}$) and identity gates on logical qubits at the $(l-1)$-th level, and is transversal at level $l-1$. For any $B_1\subseteq[k_1],\dots,B_{l}\subseteq[k_l],A_{l+1}\subseteq[n_{l+1}],\dots,A_L\subseteq[n_L]$, $U_P^{(l)}\left(B_1,\dots,B_l,A_{l+1},\dots,A_L\right)$ can be written as a tensor product of $X^{(l)}$, $Y^{(l)}$, or $Z^{(l)}$ and identity gates, which is a multiplication of tensor products of each $X^{(l)}$ ($Z^{(l)}$) and identity gates on other logical qubits at the $l$-th level. By \cref{prop:product_of_transversal}, $U_P^{(l)}\left(B_1,\dots,B_l,A_{l+1},\dots,A_L\right)$ is also transversal at level $l-1$. With appropriate choices of symplectic bases for $\mathcal{Q}_{l-1},\dots,\mathcal{Q}_{1}$, we can apply the same argument repeatedly and show that $U_P^{(l)}\left(B_1,\dots,B_l,A_{l+1},\dots,A_L\right)$ is transversal at levels $l-2,\dots,0$.
\end{proof}

\cref{thm:main2} implies that any logical Pauli operator at the $l$-th level, which is transversal at the $l$-th level by definition, can always be made transversal at any level lower than $l$. This result is quite straightforward since for any stabilizer code, there exists a symplectic basis such that any logical Pauli operator can be implemented by transversal physical gates.

\begin{theorem} \label{thm:main3}
Let $\mathcal{Q}_\mathrm{con}^{(L)}=\mathcal{Q}_L \circ \cdots \circ \mathcal{Q}_1$ be an $L$-level concatenated code with parameter \codepar{N,K,D} as defined in \cref{def:concat_code}, and suppose that any code $\mathcal{Q}_i$ ($i \in [L]$) is a self-dual CSS code satisfying the condition in \cref{thm:main1} and logical Pauli operators of each code are defined by a compatible symplectic basis. Then, for any $l,m \in \{0,\dots,L\}$ where $m < l$ and for any $j_1 \in [k_1],\dots,j_m\in [k_m],i_{l+1} \in [n_{l+1}],\dots,i_L \in [n_L]$, the following operators are transversal at levels $m,\dots,l$.
\begin{widetext}
\begin{enumerate}
    \item $U_H^{(l)}\left(\{j_1\},\dots,\{j_m\},[k_{m+1}],\dots,[k_l],\{i_{l+1}\},\dots,\{i_L\}\right)$
    \item $U_S^{(l)}\left(\{j_1\},\dots,\{j_m\},[k_{m+1}],\dots,[k_l],\{i_{l+1}\},\dots,\{i_L\};\mathbf{a}\right)$ for any $\mathbf{a} \in \{1,-1\}^c$, $c = |\{j_1\} \times \dots \times \{j_m\} \times [k_{m+1}] \times \dots \times [k_l] \times \{i_{l+1}\} \times \dots \times \{i_L\}|$.
    \item $U_\mathrm{CNOT}^{(l)}\left(\{j_1\},\dots,\{j_m\},[k_{m+1}],\dots,[k_l],\{i_{l+1}\},\dots,\{i_L\}\right)$.
\end{enumerate}
\end{widetext}
\end{theorem}

\begin{proof}
    We start by proving the statement for Hadamard-type transversal gates.
    Consider any $l \in \{0,\dots,L\}$. $\mathcal{Q}_l$ encodes logical qubits at the $(l-1)$-th level with indices in $\{j_1\}\times\cdots\times\{j_{l-1}\}\times[n_l]\times\{i_{l+1}\}\times\cdots\times\{i_L\}$ to logical qubits at the $l$-th level with indices in $\{j_1\}\times\cdots\times\{j_{l-1}\}\times[k_l]\times\{i_{l+1}\}\times\cdots\times\{i_L\}$ for any $j_1\in[k_1],\dots,j_{l-1}\in[k_{l-1}],i_{l+1}\in[n_{l+1}],\dots,i_L\in[n_L]$. Given the facts that $\mathcal{Q}_l$ is a self-dual CSS code satisfying \cref{thm:main1} and its logical Pauli operators are defined by a compatible symplectic basis, we have that for any $l$, $U_H^{(l)}\left(\{j_1\},\dots,\{j_{l-1}\},[k_l],\{i_{l+1}\},\dots,\{i_L\}\right)=U_H^{(l-1)}\left(\{j_1\},\dots,\{j_{l-1}\},[n_l],\{i_{l+1}\},\dots,\{i_L\}\right)$ for any $j_1\in[k_1],\dots,j_{l-1}\in[k_{l-1}],i_{l+1}\in[n_{l+1}],\dots,i_L\in[n_L]$; the operator is transversal at levels $l-1$ and $l$.
    Let $m \in \{0,\dots,L\}$ where $m < l$. Multiplying $U_H^{(l)}\left(\{j_1\},\dots,\{j_{l-1}\},[k_l],\{i_{l+1}\},\dots,\{i_L\}\right)$ with all possible $j_{m+1},\dots,j_{l-1}$ and using the facts that logical Pauli operators of $\mathcal{Q}_{m+1},\dots,\mathcal{Q}_l$ are defined by a compatible symplectic basis, we have that,
\begin{widetext}
\begin{align}
    &U_H^{(l)}\left(\{j_1\},\dots,\{j_m\},[k_{m+1}],\dots,[k_l],\{i_{l+1}\},\dots,\{i_L\}\right) \nonumber \\
    &=U_H^{(l-1)}\left(\{j_1\},\dots,\{j_m\},[k_{m+1}],\dots,[k_{l-1}],[n_l],\{i_{l+1}\},\dots,\{i_L\}\right) \nonumber \\
    &=\dots \nonumber \\
    &=U_H^{(m)}\left(\{j_1\},\dots,\{j_m\},[n_{m+1}],\dots,[n_l],\{i_{l+1}\},\dots,\{i_L\}\right)
\end{align}
\end{widetext}
That is, for any $j_1 \in [k_1]$, \dots, $j_m\in [k_m]$, $i_{l+1} \in [n_{l+1}]$, \dots, $i_L \in [n_L]$, $U_H^{(l)}(\{j_1\},\dots,\{j_m\},[k_{m+1}],\dots,[k_l],$ $\{i_{l+1}\},\dots,\{i_L\})$ is transversal at levels $m,\dots,l$.

The statements for phase-type and CNOT-type transversal gates can be proved using similar ideas. Note that in the case of phase-type transversal gates, a phase-type gate applied to each logical qubit at the $l$-th level in the support of $U_S^{(l)}\left(\{j_1\},\dots,\{j_m\},[k_{m+1}],\dots,[k_l],\{i_{l+1}\},\dots,\{i_L\};\mathbf{a}\right)$ can be either $S^{(l)}$ or $S^{\dagger(l)}$, and a transversal decomposition of the operator in terms of $S^{(l')}$ and $S^{\dagger(l')}$ exists at any level $l' \in \{m,\dots,l\}$. This is possible by the definition of a compatible symplectic basis (\cref{def:compatible_basis}).
\end{proof}

\cref{thm:main3} states that for Hadamard-type, phase-type, or CNOT-type transversal gate at the $l$-th level, if the gate acts on all qubits in the $(m+1)$-th to $l$-th dimensions ($m<l$), then the gate can be made transversal at levels $m,\dots,l$. For example, consider the concatenated code in \cref{fig:code_concat} and the gate $U_H^{(1)}\left([7],\{1\} \right)$ that acts on qubits on the first row across all columns (the yellow box) of the middle grid. The gate is equivalent to $U_H^{(0)}\left([15],\{1\} \right)$ that acts on qubits on the first row across all columns of left grid. Similarly, consider the gate $U_H^{(2)}\left(\{1\},[2] \right)$ that acts on qubits on the first column across all rows (the purple box) of the right grid. The gate is equivalent to $U_H^{(1)}\left(\{1\},[6] \right)$ that acts on qubits on the first column across all rows of the middle grid.

In a special case where the Hadamard-type, phase-type, or CNOT-type transversal gate acts on all logical qubits at the $L$-th level (the top level), \cref{thm:main3} can be simplified to the following corollary.

\begin{corollary} \label{cor:trans_logical_physical}
    Let $\mathcal{Q}_\mathrm{con}^{(L)}=\mathcal{Q}_L \circ \cdots \circ \mathcal{Q}_1$ be an $L$-level concatenated code with parameter \codepar{N,K,D} as defined in \cref{def:concat_code}, and suppose that any code $\mathcal{Q}_i$ ($i \in [L]$) is a self-dual CSS code satisfying the condition in \cref{thm:main1} and logical Pauli operators of each code are defined by a compatible symplectic basis. Then, 
    \begin{enumerate}
        \item $\bigotimes_{j=1}^{K} H^{(L)}_{j}$ is transversal at levels $0,\dots,L$ and can be implemented by $\bigotimes_{i=1}^{N} H^{(0)}_{i}$.
        \item For any $(a_1,\dots,a_{K}) \in \{1,-1\}^K$, $\bigotimes_{j=1}^{K} \left(S_{j}^{(L)}\right)^{a_j}$ is transversal at levels $0,\dots,L$ and can be implemented by $\bigotimes_{i=1}^{N} \left(S^{(0)}_{i}\right)^{b_i}$ for some $(b_1,\dots,b_{N}) \in \{-1,1\}^{N}$.
    \end{enumerate}
\end{corollary}

Consider a certain level $l$ and suppose that logical gates at the $l$-th level from \cref{thm:main3} can be transversally implemented by logical gates at the $(l-1)$-th level or lower. We note that these gates alone cannot generate the full logical Clifford group at the $l$-th level, similar to what was previously discussed in \cref{sec:definitions}. This is because each of these gates simultaneously applies the same Clifford operation to many logical qubits at the $l$-th level. To achieve the full logical Clifford group, additional addressable logical gates at the $l$-th level such as the ones in \cref{prop:full_Clifford} are required.

\section{Applications of compatible symplectic bases and Clifford gates with multilevel transversality} \label{sec:applications}

So far, we have proposed a concept of compatible symplectic basis for a self-dual CSS code, a basis of logical $X$ and logical $Z$ operators such that logical Clifford operators $\bigotimes_{j=1}^{k} \bar{H}_{j}$ and $\bigotimes_{j=1}^{k} \bar{S}_{j}^{a_j}$ for any $a_j \in \{-1,1\}$ can be fault-tolerantly implemented by corresponding transversal gates at the physical level, and show that a compatible symplectic basis exists if the code satisfies \cref{thm:main1}. With implementations of some additional addressable logical gates through a process such as the teleportation-based FTQC scheme in \cref{app:gate_TP}, the full logical Clifford group of the code can be achieved. We also extend the ideas to code concatenation and show that if a concatenated code is constructed from self-dual CSS codes such that each of them has a compatible symplectic basis, then certain logical gates at each level of the concatenated code can be implemented transversally by logical gates at a lower level. 

In this section, we demonstrate some applications of compatible symplectic bases and Clifford gates with multilevel transversality in FTQC. It should be noted that the purpose of this section is to discuss examples of situations in which choosing a compatible symplectic basis over an arbitrary symplectic basis could lead to simpler resource requirements. However, the comparison of our proposals to other existing methods in terms of space-time overhead requires careful simulations and quantitative analysis, which is out of scope of this work.

\subsection{Explicit construction of symplectic bases for Yamasaki-Koashi FTQC scheme on a concatenated quantum Hamming code}

In Ref. \cite{YK24}, Yamasaki and Koashi have proposed a time-efficient constant-space-overhead FTQC scheme which is operated on a concatenated quantum Hamming code. The code is constructed from concatenating \codepar{2^m-1,2^m-1-m,3} quantum Hamming codes with growing $m$. Their construction of the FTQC scheme and the proofs of fault tolerance relies on the assumption that logical Hadamard gates $\bigotimes_{j=1}^{k} \bar{H}_{j}$ can be implemented by transversal physical Hamadard gates $\bigotimes_{i=1}^{n} H_{i}$. It is obvious that $\bigotimes_{i=1}^{n} H_{i}$ preserves the stabilizer group, thus it is a logical operator of some type. However, its corresponding logical operation is not obvious and depends on the symplectic basis that define logical $X$ and $Z$ operators, as we have demonstrated with an example in \cref{sec:qHamming}. The symplectic bases on which their work were not specified in Ref. \cite{YK24}.

The main results from this work could complement the FTQC scheme in Ref. \cite{YK24}. In particular, since any \codepar{2^m-1,2^m-1-m,3} quantum Hamming code is a self-dual CSS code of odd length, by \cref{cor:odd_n}, a compatible symplectic basis of the code exists, and the procedure provided in \cref{app:python} can be used to explicitly constructed a compatible symplectic basis for any quantum Hamming code. Therefore, if logical Pauli operators for a quantum Hamming code at any level are defined with a compatible symplectic basis, the operator $\bigotimes_{j=1}^{k} \bar{H}_{j}$ of the code can be implemented by $\bigotimes_{i=1}^{n} H_{i}$, and the FTQC scheme in Ref. \cite{YK24} could operate as it is originally intended.

We point out that our results also provide possibilities to extend the FTQC scheme in Ref. \cite{YK24} to a broader family of concatenated codes; Any $\mathcal{Q}_i$, $i \in [L]$ of a concatenated code $\mathcal{Q}_\mathrm{con}^{(L)} = \mathcal{Q}_L \circ \cdots \circ \mathcal{Q}_1$ could be any self-dual CSS code satisfying \cref{thm:main1} because $\bigotimes_{j=1}^{k} \bar{H}_{j}$ on each $\mathcal{Q}_i$ can be implemented by $\bigotimes_{i=1}^{n} H_{i}$ and the FTQC scheme still works. However, the scaling of space and time overhead depends on the family of codes being used in the concatenation, and the results may differ from those of concatenated quantum Hamming codes. We also point out possibilities that the FTQC scheme in Ref. \cite{YK24} could be optimized since for any $\mathcal{Q}_i$, $\bigotimes_{j=1}^{k} \bar{S}_{j}^{a_j}$ for any $a_j \in \{-1,1\}$ could be implemented by transversal physical $S$ and $S^\dagger$ gates. We leave this optimization for future work.

\subsection{Gate conversion and optimization for concatenated self-dual CSS codes} \label{subsec:app_concat}

Consider a concatenated code $\mathcal{Q}_\mathrm{con}^{(L)} = \mathcal{Q}_L \circ \cdots \circ \mathcal{Q}_1$ and suppose that any $\mathcal{Q}_l$, $l \in [L]$ is a self-dual CSS code satisfying \cref{thm:main1}. By \cref{thm:main2,thm:main3}, we know that various logical gates of $\mathcal{Q}_\mathrm{con}^{(L)}$ are transversal at multiple levels; i.e., for such a logical gate $U^{(l)}$ acting on logical qubits at the $l$-th level, there exist multiple transversal decompositions, each in terms of gates acting on logical qubits at a level lower than $l$. A gate operating at a lower level is more favorable compared to a gate operating at a higher level as the former one has lower implementation complexity; a preparation of an ancilla state required for implementing a gate at a lower level is generally simpler and consume less space and time overhead to ensure fault tolerance. Moreover, a gate at a lower level is also better at error confinement; if faults occur in the implementation of a gate at a lower level, logical errors are limited to the level of implementation and could be removed with error correction gadgets operating on the same level, which also consume less space and time overhead compared to ones for a higher level. Therefore, if a logical gate is transversal at multiple levels, in general, we would like to implement it with gates operating on the lowest possible level. One example that can support this conjecture is the comparison between an implementation of $H^{\otimes n}$ on a (non-concatenated) self-dual CSS code versus an implementation of $\bar{H}^{\otimes k}$ by gate teleportation. When the preparation of ancilla states required for implementing $\bar{H}^{\otimes k}$ are also taken into account, we expect $H^{\otimes n}$ to be more resource-efficient and less prone to errors. We also expect similar situations to happen in the case of concatenated codes.

The fact that a compatible symplectic basis exists for each $\mathcal{Q}_l$ provide us many possible ways to optimize circuits for implementing certain logical gates. We provide below some examples of gate identities which could help converting different logical gates and could possibly minimize the number of gates that require additional process for implementation (such as the FTQC scheme in \cref{app:gate_TP}), reducing the required space and time overhead. Note that what we provide is not an exhaustive list of possible conversions. For simplicity of the description, here we treat an \codepar{n_l,k_l,d_l} code $\mathcal{Q}_l$ as a 1-level concatenated code, so in a compatible symplectic basis, $U_H^{(l)}([k_l])=\bigotimes_{j=1}^{k_l} H^{(l)}_{j}$ can be implemented by $\bigotimes_{i=1}^{n_l} H^{(l-1)}_{i}$, and for any $\mathbf{a} \in \{-1,1\}^{k_l}$, $U_S^{(l)}([k_l];\mathbf{a})=\bigotimes_{j=1}^{k_l} \left(S_{j}^{(l)}\right)^{a_j}$ can be implemented by $\bigotimes_{i=1}^{n_l} \left(S^{(l-1)}_{i}\right)^{b_i}$ for some $(b_1,\dots,b_{n_l}) \in \{-1,1\}^{n_l}$. Examples of gate conversions are shown below.

\begin{enumerate}
    \item For any $A \subseteq [k_l]$, 
    \begin{align}
        U_H^{(l)}(A) &= U_H^{(l)}([k_l])U_H^{(l)}([k_l]\setminus A) \nonumber \\
        &= \left(\bigotimes_{i=1}^{n_l} H^{(l-1)}_{i}\right)U_H^{(l)}([k_l]\setminus A).
    \end{align}
    \item For any $A \subseteq [k_l]$, 
    \begin{align}
        U_S^{(l)}(A;\mathbf{1}) &= U_S^{(l)}([k_l];\mathbf{1})U_S^{(l)}([k_l]\setminus A;-\mathbf{1}) \nonumber \\
        &= \left(\bigotimes_{i=1}^{n_l}\left(S^{(l-1)}_{i}\right)^{b_i} \right) U_S^{(l)}([k_l]\setminus A;-\mathbf{1}),
    \end{align}
    for some $(b_1,\dots,b_{n_l}) \in \{-1,1\}^{n_l}$.
    \item For any $j \in [k_l]$, 
    \begin{align}
        H_j^{(l)} S_j^{(l)} H_j^{(l)} &= U_H^{(l)}([k_l])S_j^{(l)}U_H^{(l)}([k_l]) \nonumber \\
        &= \left(\bigotimes_{i=1}^{n_l} H^{(l-1)}_{i}\right) S_j^{(l)} \left(\bigotimes_{i=1}^{n_l} H^{(l-1)}_{i}\right).
    \end{align}
\begin{widetext}      
    \item For any $j \in [k_l]$, 
    \begin{align}
        H_j^{(l)} S_j^{(l)} H_j^{(l)} &= S_j^{\dagger(l)} H_j^{(l)}S_j^{\dagger(l)} \nonumber \\
        &= U_S^{(l)}([k_l];-\mathbf{1}) H_j^{(l)} U_S^{(l)}([k_l];(1,\dots,1,\underbracket{-1}_\text{$j$-th},1,\dots,1)) \nonumber \\
        &= \left(\bigotimes_{i=1}^{n_l}\left(S^{(l-1)}_{i}\right)^{b_i} \right) H_j^{(l)} \left(\bigotimes_{i=1}^{n_l}\left(S^{(l-1)}_{i}\right)^{c_i} \right),
    \end{align}
    for some $(b_1,\dots,b_{n_l}),(c_1,\dots,c_{n_l}) \in \{-1,1\}^{n_l}$.
    \item For any $j \in [k_l]$, 
    \begin{align}
    S_j^{(l)} &= H_j^{(l)} S_j^{\dagger(l)} H_j^{(l)} S_j^{\dagger(l)} H_j^{(l)} \nonumber \\
    &= U_H^{(l)}([k_l])U_S^{(l)}([k_l];-\mathbf{1}) H_j^{(l)} U_S^{(l)}([k_l];(1,\dots,1,\underbracket{-1}_\text{$j$-th},1,\dots,1)) U_H^{(l)}([k_l]) \nonumber \\
    &= \left(\bigotimes_{i=1}^{n_l} H^{(l-1)}_{i}\right) \left(\bigotimes_{i=1}^{n_l}\left(S^{(l-1)}_{i}\right)^{b_i} \right) H_j^{(l)} \left(\bigotimes_{i=1}^{n_l}\left(S^{(l-1)}_{i}\right)^{c_i} \right) \left(\bigotimes_{i=1}^{n_l} H^{(l-1)}_{i}\right),
    \end{align}
    for some $(b_1,\dots,b_{n_l}),(c_1,\dots,c_{n_l}) \in \{-1,1\}^{n_l}$.
\end{widetext}
\end{enumerate}

Here we use the identity $HSHSHS = I$ up to some global phase. Circuit diagrams corresponding the examples above are provided in \cref{fig:opt_tricks_1} (where qubits and gates in each diagram correspond to logical gates and logical qubits of the code).

\begin{figure*}[!htbp]
  \centering
  \begin{subfigure}[b]{0.35\textwidth}
    \centering
    \includegraphics[scale=1]{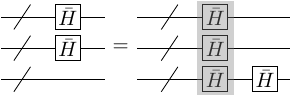}
    \captionsetup{justification=centering}
    \caption{}
  \end{subfigure}
  \vspace{2mm}
  \begin{subfigure}[b]{0.35\textwidth}
    \centering
    \includegraphics[scale=1]{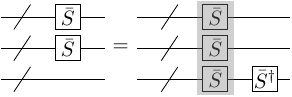}
    \captionsetup{justification=centering}
    \caption{}
  \end{subfigure}
  \vspace{2mm}
  \begin{subfigure}[b]{0.9\textwidth}
    \centering
    \includegraphics[scale=1]{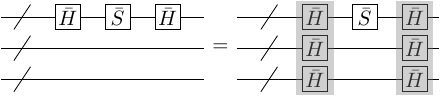}
    \captionsetup{justification=centering}
    \caption{}
  \end{subfigure}
  \vspace{2mm}
  \begin{subfigure}[b]{0.9\textwidth}
    \centering
    \includegraphics[scale=1]{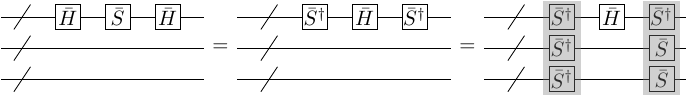}
    \captionsetup{justification=centering}
    \caption{}
  \end{subfigure}
  \vspace{2mm}
  \begin{subfigure}[b]{0.9\textwidth}
    \centering
    \includegraphics[scale=1]{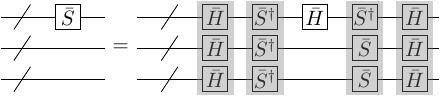}
    \captionsetup{justification=centering}
    \caption{}
  \end{subfigure}
	\caption{Examples of possible gate conversions given that Hadamard-type and phase-type logical-level transversal gates can be implemented by physical-level transversal gates. Diagrams (a)-(e) correspond to Examples 1-5 in \cref{subsec:app_concat}, respectively. A shaded block indicates a logical operator of the form $\bar{H}^{\otimes k}$ (or $\bigotimes_{j=1}^{k} \bar{S}_j^{a_j}$ for some $(a_1,\dots,a_k)\in\{-1,1\}^k$) which can be implemented by $H^{\otimes n}$ (or $ \bigotimes_{i=1}^{n} S_i^{b_i}$ for some $(b_1,\dots,b_n)\in\{-1,1\}^n$).}
	\label{fig:opt_tricks_1}%
\end{figure*}

In Examples 1 and 2, the number of gates at the $l$-th level to be implemented is $k_l-|A|$ instead of $|A|$, which could be smaller than $|A|$ if $|A| > k_l/2$. In Examples 3, 4, and 5, the number of gates at the $l$-th level to be implemented is reduced to 1 from 3, 3, and 5, respectively. In all examples, by inserting appropriate gates at the $(l-1)$-th level, one might be able to further convert all gates at the $(l-1)$-th level to gates at the physical level. These examples demonstrate possibilities to optimize space and time overhead when a quantum code has a compatible symplectic basis.

Furthermore, it is possible to show that in some cases, a product of two logical gates which are transversal at level $m$ could give a logical gate which is transversal at level $m' < m$. As an example, let $\mathcal{Q}_\mathrm{con}^{(L)} = \mathcal{Q}_2 \circ \mathcal{Q}_1$ be a concatenated code where $\mathcal{Q}_1$ is the \codepar{15,7,3} code and $\mathcal{Q}_2$ is the \codepar{6,2,2} code, and consider a product of two gates $U_1$ and $U_2$ whose supporting logical qubits are displayed in \cref{fig:opt_tricks_2}. $U_1$ is a Hadamard-type transversal gate at the second level, which is also transversal at the first level (by \cref{thm:main3}). Meanwhile, $U_2$ is a Hadamard-type transversal gate which is transversal at the first level only. The product $U_1 U_2$ is a Hadamard-type transversal gate which is transversal at the zeroth and the first levels. As a result, $U_1 U_2$ can be implemented using only physical gates, and additional processes for implementing high-level logical gates are not required. In this example, the conversion to zeroth-level transversal gates is possible because $U_1$ and $U_2$ are of a special form: At the first level, $U_2$ cannot be further converted to zeroth-level transversal gates because it operates on too many logical qubits in some region and lacks support on some region. Meanwhile, $U_1$ operates on both of the aforementioned regions of $U_2$ and does not operate on any other regions. Therefore, $U_1$ fills in the missing part and cancels the excessive part of $U_2$, making further conversion of $U_1 U_2$ possible.

\begin{figure*}[!htbp]
	\centering
	\includegraphics[width=0.75\textwidth]{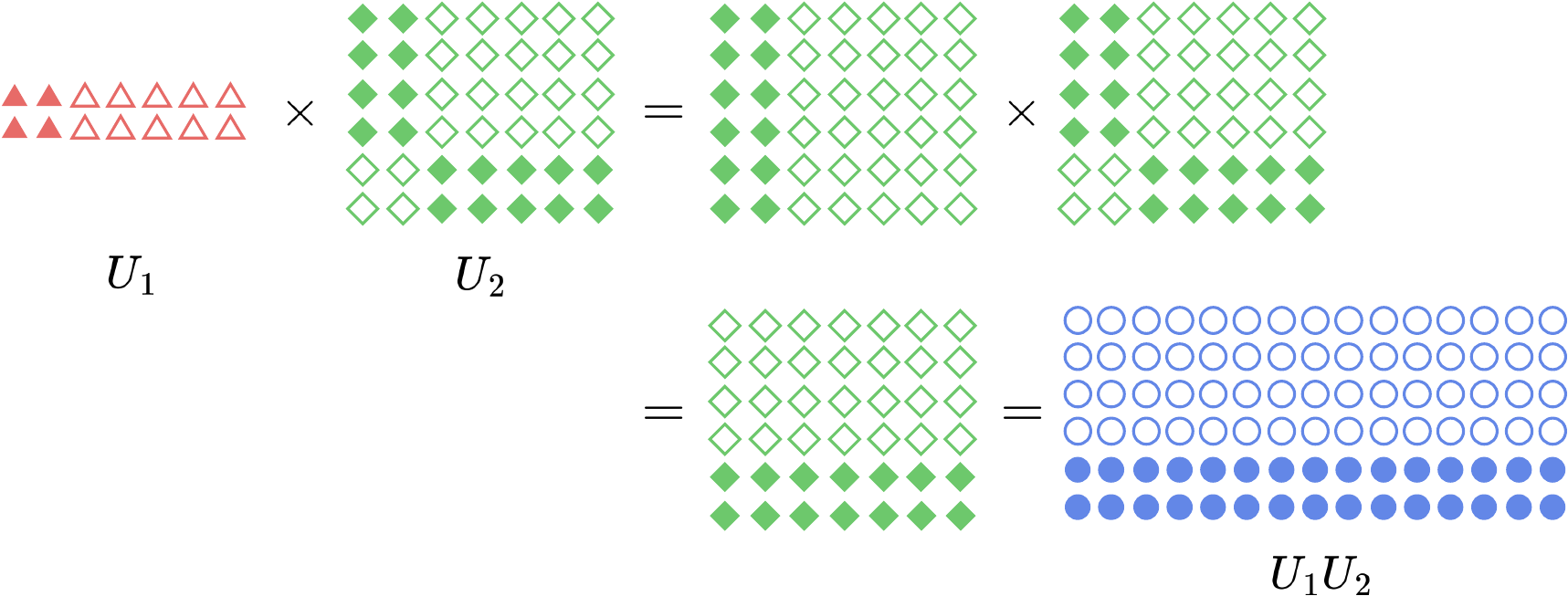}
	\caption{A product of two logical gates which are transversal at level $m$ could give a logical gate which is transversal at level $m' < m$. In this example, a concatenated code $\mathcal{Q}_2 \circ \mathcal{Q}_1$ is constructed from the \codepar{15,7,3} code $\mathcal{Q}_1$ and the \codepar{6,2,2} code $\mathcal{Q}_2$. Here, $U_1$ is a Hadamard-type transversal gate at the second level acting on the filled triangles (second-level logical qubits), and $U_2$ is a Hadamard-type transversal gate at the first level acting on the filled diamonds (first-level logical qubits). In this diagram, $\times$ refers to a multiplication of unitary operators, and $=$ determines the logical equivalence of two unitary operators. This diagram shows that the product of $U_1$ and $U_2$ is a logical gate that can be implemented by transversal physical Hadamard gates, i.e., Hadamard gates acting on the filled circles (physical qubits).}
	\label{fig:opt_tricks_2}%
\end{figure*}

We point out that the situations similar to our example where high-level and low-level logical gates are applied consecutively are not uncommon. In an actual implementation of logical gates on a concatenated code, each logical gate at the $l$-th level should be accompanied by QEC gadgets to control the error propagation \cite{AGP06}. Such QEC gadgets are implemented by logical gates at the $(l-1)$-th level, which are also accompanied by error correction gadgets at a lower level. In such situations, gate conversions according to \cref{thm:main2,thm:main3} could lead to potential savings in space and time overhead. 

\subsection{Resource simplification for teleportation-based fault-tolerant quantum computation} \label{subsec:app_TPQC}

The teleportation-based FTQC scheme proposed by Brun et al. \cite{BZHJL15}, which we review in \cref{app:gate_TP}, can be used to perform any logical Clifford gate on a single block of an \codepar{n,k,d} CSS code. The scheme utilizes Steane FTM scheme \cite{Steane99}, so each logical Pauli measurement requires a specific ancilla state, which can be prepared fault-tolerantly and efficiently by the method proposed in Ref. \cite{ZLB18}. Logical Pauli measurements required to perform basic logical Clifford gates $\bar{H}_i$, $\bar{S}_i$, $\overline{\mathrm{CNOT}}_{ij}$, and $\overline{\mathrm{SWAP}}_{ij}$ on any logical qubit (or any pair of logical qubits) on a block can be described by an overcomplete set $\{\bar{X}_{i},\bar{Y}_{i},\bar{Z}_{i},\bar{X}_{i}\bar{X}_{j},\bar{Y}_{i}\bar{Y}_{j},\bar{Z}_{i}\bar{Z}_{j},\bar{X}_{i}\bar{Y}_{j},\bar{Y}_{i}\bar{Z}_{j},\bar{X}_{i}\bar{Z}_{j},$ $\bar{X}_{i}\bar{X}_{j}\bar{X}_{l},\bar{Z}_{i}\bar{Z}_{j}\bar{Z}_{l}\}_{i,j,l\in[k]}$.

If the code being used is a self-dual CSS code satisfying \cref{thm:main1}, our main results could simplify the types of required ancilla states. In particular, suppose that a logical Pauli operator $\bar{P}$ (which can act on multiple logical qubits) is to be measured. Measuring $\bar{P}$ on an input state $\ket{\bar{\psi}}$ corresponds to projecting the state with projection operators $\Pi_{\pm}\!\left(\bar{P}\right) = \frac{\mathbb{1} \pm \bar{P}}{2}$. The resulting state is either $\left(\frac{\mathbb{1} + \bar{P}}{2}\right)\ket{\bar{\psi}}$ or $\left(\frac{\mathbb{1} - \bar{P}}{2}\right)\ket{\bar{\psi}}$, depending on the measurement outcome. Let $\bar{U}$ be a unitary operator, and suppose that we apply $\bar{U}^\dagger$ to the input state $\ket{\bar{\psi}}$ and apply $\bar{U}$ to the state after measurement, as illustrated in \cref{fig:meas_conversion}. Then, the resulting state is,
\begin{align}
    \bar{U}\Pi_{\pm}\!\left(\bar{P}\right)\bar{U}^\dagger\ket{\bar{\psi}} &= \bar{U} \left(\frac{\mathbb{1} \pm \bar{P}}{2}\right) \bar{U}^\dagger\ket{\bar{\psi}} \nonumber \\
    &= \frac{\mathbb{1} \pm \bar{U}\bar{P}\bar{U}^\dagger}{2} \ket{\bar{\psi}} \nonumber \\
    &= \Pi_{\pm}\!\left(\bar{U}\bar{P}\bar{U}^\dagger\right)\ket{\bar{\psi}}.
\end{align}
That is, applying $\bar{U}^\dagger$ before and $\bar{U}$ after measuring $\bar{P}$ is equivalent to measuring an operator $\bar{U}\bar{P}\bar{U}^\dagger$.

Suppose that the self-dual CSS code being used has a compatible symplectic basis. By definition, $\bar{H}^{\otimes k}$ can be implemented by $H^{\otimes n}$, and $\bar{S}^{\otimes k}$ can be implemented by some combination of physical $S$ and $S^\dagger$ gates. The implementations of these operators are fault tolerant, can be done with constant time overhead, and require no additional ancilla qubits. Thus, we can use these transversal logical operators to convert logical Pauli operators to be measured fault-tolerantly and efficiently. With this kind of conversion, the same ancilla state can be used for multiple types of logical Pauli measurements. For example, $\bar{Z}_i$ ($\bar{Y}_i$) can be measured using the same ancilla state as $\bar{X}_i$ by using $U=\bar{H}^{\otimes k}$ ($U=\bar{S}^{\otimes k}$) in the conversion. Transformations for two- and three-qubit logical Pauli measurements are described below (up to $\pm 1$ phase).
\begin{widetext}
\begin{equation}
    \bar{Z}_i\bar{Z}_j \xleftrightarrow{\bar{H}^{\otimes k}}  \bar{X}_i\bar{X}_j \xleftrightarrow{\bar{S}^{\otimes k}} \bar{Y}_i\bar{Y}_j 
\end{equation}
\begin{equation}
    \bar{Z}_i\bar{Z}_j\bar{Z}_l \xleftrightarrow{\bar{H}^{\otimes k}}  \bar{X}_i\bar{X}_j\bar{X}_l \xleftrightarrow{\bar{S}^{\otimes k}} \bar{Y}_i\bar{Y}_j\bar{Y}_l  
\end{equation}
\begin{equation}
    \bar{X}_i\bar{Z}_j \xleftrightarrow{\bar{S}^{\otimes k}} \bar{Y}_i\bar{Z}_j \xleftrightarrow{\bar{H}^{\otimes k}} \bar{Y}_i\bar{X}_j \xleftrightarrow{\bar{S}^{\otimes k}} \bar{X}_i\bar{Y}_j \xleftrightarrow{\bar{H}^{\otimes k}} \bar{Z}_i\bar{Y}_j \xleftrightarrow{\bar{S}^{\otimes k}} \bar{Z}_i\bar{X}_j
\end{equation}
\end{widetext}
We believe that this simplification of types of required ancilla states could be helpful in the architectural design of quantum computers based on logical teleportation.

\begin{figure}[!bp]
	\centering
	\includegraphics[scale=1]{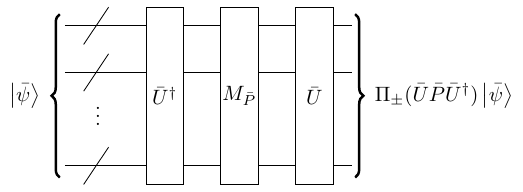}
	\caption{Inserting $\bar{U}^\dagger$ before and $\bar{U}$ after measuring a logical Pauli operator $\bar{P}$ is equivalent to measuring a logical operator $\bar{U}\bar{P}\bar{U}^\dagger$.}
	\label{fig:meas_conversion}%
\end{figure}

It should be noted that any logical Clifford circuit could be implemented fault-tolerantly in constant time by the teleportation-based FTQC scheme proposed in Ref. \cite{ZLBK20}. In that work, the authors proposed a way to decompose any Clifford circuit into a constant number of parts, then implement each part with some clean ancilla states. As noted by the authors, their method turns the complexity of the logical Clifford circuit into the complexity of preparation of the required ancilla states. The required ancilla state for each part of decomposition depends on the supporting qubits of gates in the part, and a fault-tolerant protocol to prepare such ancilla states are also provided in the same work. We point out that it is possible that in some special cases, our implementation of $\bar{H}^{\otimes k}$ and $\bar{S}^{\otimes k}$ could be used to simplify the required ancilla states. However, this application might not be possible in a general case since gates from different parts can act on different sets of logical qubits.

\section{Discussion and conclusions} \label{sec:conclusion}

Transversal implementation of logical gates is a desirable property for any QECC as it provides a way to perform FTQC with constant time overhead and without additional ancilla qubits. On a QECC that encodes more than one logical qubit, the action of certain transversal physical gates could depend on the symplectic basis that define logical Pauli operators of the code. In this work, we study self-dual CSS codes and try to find symplectic basis such that logical Hadamard gates and logical phase gates of any form can be implemented by transversal physical gates of similar types, which we call a compatible symplectic basis (\cref{def:compatible_basis}). If such a basis exists, any logical Clifford gate that acts on all logical qubits in a code block simultaneously (including $\bar{H}^{\otimes k}$ and $\bar{S}^{\otimes k}$ on a single code block and $\overline{\mathrm{CNOT}}^{\otimes k}$ between two code blocks) can be implemented by transversal physical gates.

Our first main result is necessary and sufficient conditions for any \codepar{n,k,d} self-dual CSS code with $k \geq 1$ to have a compatible symplectic basis in \cref{thm:main1}. We also provide a procedure to construct such a compatible symplectic basis of the code if it exists in \cref{app:python}. In addition, we show in \cref{cor:odd_n} that any \codepar{n,k,d} self-dual CSS code with $k \geq 1$ and odd $n$ satisfy the necessary and sufficient conditions. The family of codes includes any quantum Hamming code and any binary qBCH code provided in Ref. \cite{GB99}. 

A direct application of our first result is a construction of symplectic bases for quantum Hamming codes which allow the time-efficient constant-space-overhead FTQC scheme by Yamasaki and Koashi \cite{YK24} to work as intended. Our results also provide possibilities to optimize their FTQC scheme using available transversal logical phase gates, and to generalize the scheme to a broader class of concatenated codes such as codes obtained by concatenating self-dual CSS codes that satisfy \cref{thm:main1}. However, whether the scheme for such codes can obtained the same scaling for space and time overhead is has yet to be proved. 

Another application of our first result is a simplification of the ancilla states required for the teleportation-based FTQC schemes by Brun et al. \cite{BZHJL15}. The FTQC scheme is based on Steane FTM scheme \cite{Steane99} which consumes different ancilla states for different logical Pauli measurements. With transversal logical Hadamard and transversal logical phase gates available, some logical Pauli measurements can be converted to one another, allowing them to use the same ancilla state. This simplification could lead to a simpler architecture of quantum computers based on logical teleportation. We believe that the fact that several transversal logical gates can be implemented by transversal physical gates could provide ways to optimizing other existing fault-tolerant protocols for QECCs with high encoding rate, or even lead to new efficient fault-tolerant protocols for such codes. 

Our second main result is an extension of the first result on transversal gates to concatenated codes, and concrete definitions of Clifford gates with multilevel transversality. We prove in \cref{thm:main2,thm:main3} that by concatenating self-dual CSS codes that satisfy the necessary and sufficient conditions in \cref{thm:main1}, certain logical Clifford gates of the resulting code can be implemented transversally by logical gates at a lower level of concatenation in many different ways. These gates include transversal logical Hadamard and transversal logical phase gates that act on all logical qubits at the top level of concatenation.

With the second result, we demonstrate several circuit identities which can be used to optimize the number of logical gates at each level of concatenation. We also show that a product of some transversal logical gates which are transversal at a certain level could lead to a transversal implementation at lower level (which is not achievable by each individual logical gate). We believe that these tools could be useful for an optimization of space and time overhead in fault-tolerant protocols for concatenated codes.

The self-dual CSS code family includes quantum codes constructed from efficient classical codes such as quantum Hamming code \cite{Steane96_qHamming}, some quantum Reed--Muller codes \cite{ADP14}, and some quantum BCH codes \cite{GB99}. The family also includes some quantum LDPC codes such as color codes \cite{BM06} and hyperbolic floquet codes \cite{HB24}, so our results are applicable to a broad class of codes that can be useful for various architectures. We note that with only transversal logical Clifford gates allowed by a compatible symplectic basis, one cannot achieve the full logical Clifford group defined on all logical qubits across all code blocks if the number of logical qubits in each code block is greater than 1 \cite{CG26}. In fact, the only symmetry that a generic self-dual CSS code has is that $X$-type and $Z$-type stabilizers are of the same form, so implementing logical gates other than $\bar{H}^{\otimes k}$ and $\bigotimes_{j=1}^{k} \bar{S}_j^{a_j}$ for some $(a_1,\dots,a_k)\in\{-1,1\}^k$ may not be possible in general.

For a generic self-dual CSS code, Some choices of additional addressable logical Clifford gates that can be used together with available transversal gates allowed by a compatible symplectic basis (if it exists) to achieve the full logical Clifford group are proposed in \cref{prop:full_Clifford}. We note that these choices of additional gates might not be optimal, and some specific families of self-dual CSS codes with richer symmetries may admit a larger set of logical gates that can be implemented through transversal gates, fold-transversal gates, or any other techniques. Readers may be interested in a work in Ref. \cite{TCT26} by one of the authors of this work, which shows that the full logical Clifford group of certain quantum Reed--Muller codes can be constructed from a generating set comprising only transversal and fold-transversal gates. We leave the study of available addressable logical gates for other families of self-dual CSS codes for future work.


\textit{Note added}: After completing our manuscript, we noticed that some parts of our necessary and sufficient conditions in \cref{thm:main1} are similar to the results independently discovered by Haah et al. \cite{HHPW17}; A normal magic basis defined in Ref. \cite{HHPW17} is a symplectic basis $\{(\bar{X}_j,\bar{Z}_j)\}_{j\in[k]}$ such that $\mathrm{supp}(\bar{X}_j)=\mathrm{supp}(\bar{Z}_j)$ for all $j \in [k]$, which is equivalent to a compatible symplectic basis in our work by \cref{thm:main1}. Meanwhile, Theorem 3.4 of Ref. \cite{HHPW17} is equivalent to our \cref{lem:alg,lem:basis_change} combined, and Theorem 3.5 of Ref. \cite{HHPW17} is similar to the proof $(3 \Rightarrow 4)$ of our \cref{thm:main1} for transversal logical Hadamard gates. However, there are several differences between our work and Ref. \cite{HHPW17}: (1) While some parts of our proofs of \cref{thm:main1} and relevant lemmas can be viewed as alternative proofs to Theorems 3.4 and 3.5 of Ref. \cite{HHPW17}, we also provide an explicit procedure to construct a compatible symplectic basis for any self-dual CSS code if it exists, which is arguably simpler than the method in the proof of Theorem 3.4 of Ref. \cite{HHPW17}. (2) In addition to transversal logical Hadamard gates from transversal physical Hadamard gates, our definition of compatible symplectic bases also considers transversal logical $S$ and $S^\dagger$ gates from transversal physical $S$ and $S^\dagger$ gates. We prove in \cref{lem:stb_preserving_S} that any self-dual CSS code has transversal physical $S$ and $S^\dagger$ gates that preserve the code space, then prove in $(3 \Rightarrow 4)$ of \cref{thm:main1} that transversal logical $S$ and $S^\dagger$ gates of any form can be constructed if the code has a compatible symplectic basis. (3) A Python implementation of our procedure to construct a compatible symplectic basis and transversal logical $S$ and $S^\dagger$ gates of any form is also provided in this work.

\section{Author contributions}

Theerapat Tansuwannont led the project, contributed to all mathematical statements and most of the applications of the results, and wrote the majority of the paper. Yugo Takada contributed to the Python implementation of the procedure for constructing a compatible symplectic basis and transversal logical phase gates for a self-dual CSS code, wrote \cref{app:python}, and made all figures in the paper. Keisuke Fujii provided helpful ideas that lead to the construction of transversal logical phase gates and the application of the results to teleportation-based fault-tolerant quantum computation, acquired the funding to support the project, and oversaw the project progress.

\section{Acknowledgements}

We thank Hayata Yamasaki and Satoshi Yoshida for useful discussions on the preliminary results for quantum Hamming codes. We also thank Victor Albert, Ken Brown, and members of Duke Quantum Center for fruitful discussions on possible future works. This work is supported by MEXT Quantum Leap Flagship Program (MEXT Q-LEAP) Grant No. JPMXS0120319794, JST COI-NEXT Grant No. JPMJPF2014, and JST Moonshot R\&D Grant No. JPMJMS2061 and JPMJMS256E.

\bibliographystyle{quantum}
\bibliography{refs}

\appendix
\onecolumngrid

\section{Proof of Proposition 1}\label{app:full_Clifford}

Let $\mathcal{Q}$ be an \codepar{n,k,d} stabilizer code and suppose that there are $m$ blocks of $\mathcal{Q}$.
Also, let $\bar{G}_{p:j}$ denote a logical single-qubit gate acting on the $j$-th logical qubit of the $p$-th code block, and let $\bar{F}_{p:j,q:l}$ denote a logical two-qubit gate acting on the $j$-th logical qubit of the $p$-th code block and the $l$-th logical qubit of the $q$-th code block. Suppose that logical gates $\bigotimes_{j=1}^{k} \bar{H}_{p,j}$ and $\bigotimes_{j=1}^{k} \bar{S}_{p,j}^{a_j}$ for any $(a_1,\dots,a_k)\in\{-1,1\}^k$ can be implemented any code block, and suppose that logical gates $\bigotimes_{j=1}^{k} \overline{\mathrm{CNOT}}_{p:j,q:j}$ can be implemented between any pair of code blocks; i.e., the set of logical gates in 1) is available. We want to show that by adding some additional gates, any logical Clifford gate acting on any logical qubits between any code blocks can be implemented; i.e., the full logical Clifford group $\bar{\mathcal{C}}_{mk}$ across all code blocks can be achieved. 

First, we show that the full logical Clifford group $\bar{\mathcal{C}}_{k}$ on any code block can be achieved by adding one set of logical gates in 2). Here we use the fact that $\mathcal{C}_{k}$ is generated by $\{H_i,S_i,\mathrm{CNOT}_{ij}\}_{i,j \in [k]}$.

Case 2.1): Assume that $\bar{H}_{p,j}$ on any logical qubit $j$ of a code block $p$ and $\overline{\mathrm{CNOT}}_{p:j,p:l}$ on any pair of logical qubits $j$ and $l$ of a code block $p$ can be implemented. We can show that $\bar{S}_{p,j}$ on any logical qubit $j$ of a code block $p$ can be implemented. The following circuit diagram describes an implementation of $\bar{S}_{p,1}$:
\begin{equation}
    \includegraphics[scale=1]{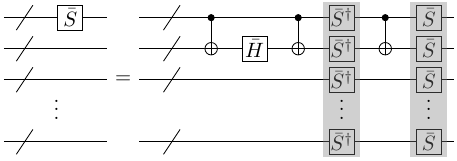}
\end{equation}
With $\bar{H}_{p,j}$, $\bar{S}_{p,j}$, and $\overline{\mathrm{CNOT}}_{p:j,p:l}$, the full Clifford group on each of the code block can be achieved.

Case 2.2): Assume that $\bar{S}_{p,j}$ on any logical qubit $j$ of a code block $p$ and $\overline{\mathrm{CNOT}}_{p:j,p:l}$ on any pair of logical qubits $j$ and $l$ of a code block $p$ can be implemented. We can show that $\bar{H}_{p,j}$ on any logical qubit $j$ of a code block $p$ can be implemented. The following circuit diagram describes an implementation of $\bar{H}_{p,1}$:
\begin{equation}
    \includegraphics[scale=1]{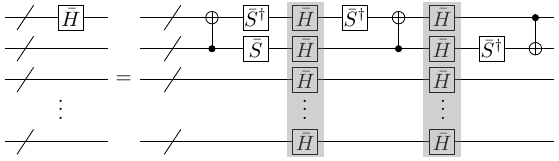}
\end{equation}
With $\bar{H}_{p,j}$, $\bar{S}_{p,j}$, and $\overline{\mathrm{CNOT}}_{p:j,p:l}$, the full Clifford group on each of the code block can be achieved.

Case 2.3): Assume that $\bar{H}_{p,j}$ on any logical qubit $j$ of a code block $p$ and $\overline{\mathrm{CZ}}_{p:j,p:l}$ on any pair of logical qubits $j$ and $l$ of a code block $p$ can be implemented. We can show that $\bar{S}_{p,j}$ on any logical qubit $j$ of a code block $p$, and $\overline{\mathrm{CNOT}}_{p:j,p:l}$ on any pair of logical qubits $j$ and $l$ of a code block $p$ can be implemented. The following circuit diagram describes an implementation of $\bar{S}_{p,1}$:
\begin{equation}
    \includegraphics[scale=1]{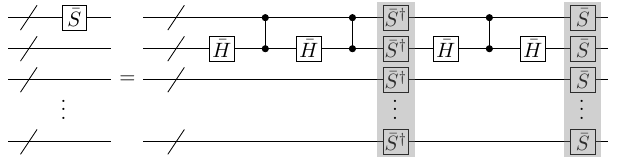}
\end{equation}
The following circuit diagram describes an implementation of $\overline{\mathrm{CNOT}}_{p:1,p:2}$:
\begin{equation}
    \includegraphics[scale=1]{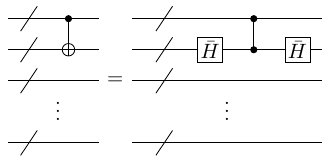}
\end{equation}
With $\bar{H}_{p,j}$, $\bar{S}_{p,j}$, and $\overline{\mathrm{CNOT}}_{p:j,p:l}$, the full Clifford group on each of the code block can be achieved.

Case 2.4): Assume that $\overline{\mathrm{CNOT}}_{p:j,p:l}$ on any pair of logical qubits $j$ and $l$ of a code block $p$, and $\overline{\mathrm{CZ}}_{p:j,p:l}$ on any pair of logical qubits $j$ and $l$ of a code block $p$ can be implemented. We can show that $\bar{H}_{p,j}$ on any logical qubit $j$ of a code block $p$, and $\bar{S}_{p,j}$ on any logical qubit $j$ of a code block $p$ can be implemented. The following circuit diagram describes an implementation of $\bar{H}_{p,1}$:
\begin{equation}
    \includegraphics[scale=1]{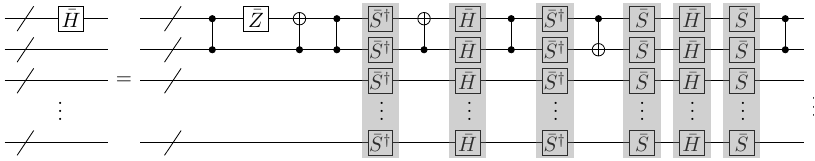}
\end{equation}
The following circuit diagram describes an implementation of $\bar{S}_{p,1}$:
\begin{equation}
    \includegraphics[scale=1]{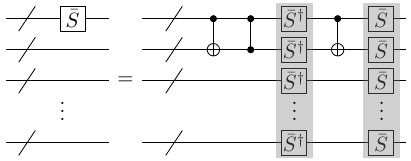}
\end{equation}
With $\bar{H}_{p,j}$, $\bar{S}_{p,j}$, and $\overline{\mathrm{CNOT}}_{p:j,p:l}$, the full Clifford group on each of the code block can be achieved.

In the next step, we show that the full logical Clifford group $\bar{\mathcal{C}}_{mk}$ across all code blocks can be achieved by adding one set of logical gates in 3). 

Case 3.1): Currently, we can implement $\bar{H}_{p,j}$ and $\bar{S}_{p,j}$ on any logical qubit of any code block, so what we need to show is that a logical CNOT gate $\overline{\mathrm{CNOT}}_{p:j,q:l}$ between any pair of logical qubits $j,l$ from any pair of code blocks $p,q$ can be implemented. Here we assume that a logical CNOT gate $\overline{\mathrm{CNOT}}_{p:1,q:1}$ can be implemented between any pair of code blocks $p,q$. Observe that a logical SWAP gate $\overline{\mathrm{SWAP}}_{p:j,p:l}$ between logical qubits $j,l$ of a code block $p$ can be implemented by $\overline{\mathrm{CNOT}}_{p:j,p:l}\overline{\mathrm{CNOT}}_{p:l,p:j}\overline{\mathrm{CNOT}}_{p:j,p:l}$. With $\overline{\mathrm{CNOT}}_{p:1,q:1}$, $\overline{\mathrm{SWAP}}_{p:1,p:j}$, and $\overline{\mathrm{SWAP}}_{q:1,q:l}$, logical gates $\overline{\mathrm{CNOT}}_{p:j,q:l}$ between any pair of logical qubits $j,l$ from any pair of code blocks $p,q$ can be implemented. Therefore, the full logical Clifford group $\bar{\mathcal{C}}_{mk}$ can be achieved.

Case 3.2): In this case, $\overline{\mathrm{CZ}}_{p:1,q:1}$ are available instead of $\overline{\mathrm{CNOT}}_{p:1,q:1}$. We can construct $\overline{\mathrm{CZ}}_{p:1,q:1}$ from $\overline{\mathrm{CNOT}}_{p:1,q:1}$ and $\bar{H}_{q,1}$. The rest of the proof is similar to the one for Case 3.1. \hfill \qedsymbol

\section{Teleportation-based fault-tolerant quantum computation}\label{app:gate_TP}

In this section, we review the teleportation-based FTQC scheme proposed by Brun et al. \cite{BZHJL15}. The scheme is based in Steane FTM scheme \cite{Steane99} and can be viewed as a simplified version of logical gate teleportation. Note that the scheme is applicable only when the code is a CSS code.

We start by reviewing Steane FTEC scheme \cite{Steane97,Steane04}. The scheme for an \codepar{n,k,d} CSS code requires two blocks of ancilla qubits. One block is prepared in the $\ket{\bar{+}}^{\otimes k}$ state and the other is prepared in the $\ket{\bar{0}}^{\otimes k}$ state, which will be called $Z$ and $X$ ancilla blocks, respectively. Transversal physical CNOT gates $\mathrm{CNOT}^{\otimes n}$ (which implement transversal logical CNOT gates $\overline{\mathrm{CNOT}}^{\otimes k}$) are applied from the data block to the $Z$ ancilla block. Also, transversal physical CNOT gates are applied from the $X$ ancilla block to the data block. Afterwards, physical qubits in the $Z$ ancilla block are transversally measured in $Z$ basis. A $Z$-type syndrome for correcting $X$-type errors can be constructed from the measurement results according to the $Z$-type generators (or $Z$ checks) on the CSS code. Similarly, physical qubits in the $X$ ancilla block are transversally measured in $X$ basis, and an $X$-type syndrome for correcting $Z$ errors can be constructed. Lastly, QEC can be done on the data block by applying Pauli operators corresponding to the $Z$- and the $X$-type syndromes.

If the $Z$ ($X$) ancilla block is prepared in the state other than $\ket{\bar{+}}^{\otimes k}$ ($\ket{\bar{0}}^{\otimes k}$), the process that we previously described can induce logical measurement on the data block; this is Steane FTM scheme \cite{Steane99}. For example, if the $i$-th logical qubit in the $Z$ ancilla block is prepared in the $\ket{\bar{0}}$ state and the other logical qubits are prepared in the $\ket{\bar{+}}$ state, and all logical qubits in the $X$ ancilla block are prepared in the $\ket{\bar{0}}$ state, then the process induces $\bar{Z}_i$ measurement on the data block. In contrast, if the $i$-th logical qubit in the $X$ ancilla block is prepared in the $\ket{\bar{+}}$ state and the other logical qubits are prepared in the $\ket{\bar{0}}$ state, and all logical qubits in the $Z$ ancilla block are prepared in the $\ket{\bar{+}}$ state, then the process induces $\bar{X}_i$ measurement on the data block. Measuring $\bar{Y}_i$ and joint logical Pauli operators such as $\bar{X}_i\bar{X}_j$, $\bar{Z}_i\bar{Z}_j$, or $\bar{X}_i\bar{Z}_j$ are possible, but the required ancilla state for some of these logical Pauli measurements must be an entangled state between $Z$ and $X$ ancilla blocks. Stabilizers that described the required ancilla state of these logical Pauli measurements can be found in Ref. \cite{BZHJL15}.

\begin{figure}[!htbp]
	\centering
	\begin{subfigure}[b]{0.4\textwidth}
        \includegraphics[height=8cm, keepaspectratio]{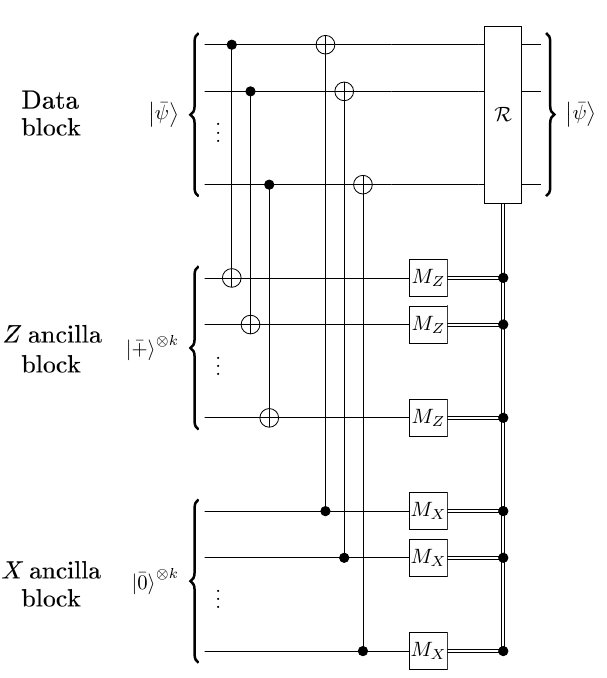}
        \captionsetup{justification=centering}
        \caption{}
    \end{subfigure}
    \begin{subfigure}[b]{0.49\textwidth}
        \includegraphics[height=8cm, keepaspectratio]{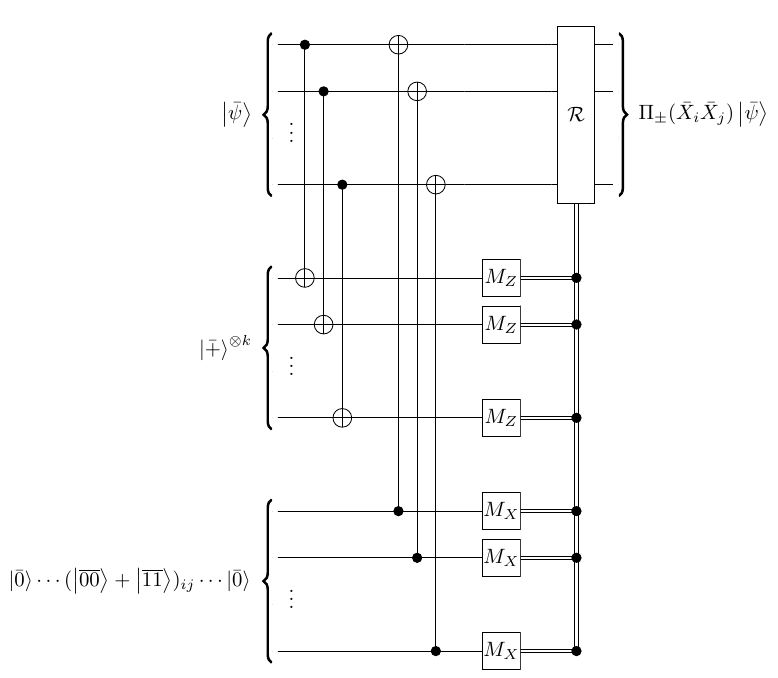}
        \captionsetup{justification=centering}
        \caption{}
    \end{subfigure}
	\caption{(a) Steane FTEC. To measure all stabilizer generators simultaneously, the $Z$ ancilla block is prepared in the $\ket{\bar{+}}^{\otimes k}$ state, and $X$ ancilla block is prepared in the $\ket{\bar{0}}^{\otimes k}$ state. Transversal physical CNOT gates are applied from the data block to the $Z$ ancilla block, and from the $X$ ancilla block to the data block. Then, physical qubits in the $Z$ ($X$) ancilla block are transversally measured in $Z$ ($X$) basis. (b) Steane FTM. Logical Pauli measurement on the code block is measured by preparing the $Z$ ($X$) ancilla block in a state other than $\ket{\bar{+}}^{\otimes k}$ ($\ket{\bar{0}}^{\otimes k}$). In this example, $\bar{X}_i\bar{X}_j$ is measured by preparing the $X$ ancilla block in the state $\ket{\overline{00}}+\ket{\overline{11}}$ between the $i$-th and the $j$-th logical qubits, and in the state $\ket{\bar{0}}$ for other logical qubits.}
    \label{fig:Steane_FTEC_FTM}%
\end{figure}

Next, we describe how FTQC on logical qubits can be done on a single block of an \codepar{n,k,d} CSS code. The scheme in Ref. \cite{BZHJL15} assumes that there are some available logical qubits that can be used as buffer qubits. Suppose that the first logical qubit is used as a buffer qubit and the input state on the remaining $k-1$ logical qubits is $\ket{\bar{\psi}}$. 
\begin{enumerate}
    \item A logical SWAP gate $\overline{\mathrm{SWAP}}_{ij}$ can be implemented by measuring the following operators: $\bar{X}_1\bar{X}_i\bar{X}_j$, followed by $\bar{Z}_1\bar{Z}_i\bar{Z}_j$, followed by $\bar{X}_1$. The first logical qubit will be left in the $\ket{\bar{+}}$ or $\ket{\bar{-}}$ state, and the remaining logical qubits will be in the state $\overline{\mathrm{SWAP}}_{ij}\ket{\bar{\psi}}$ up to some logical Pauli correction.
    \item A logical Hadamard gate $\bar{H}_i$ can be implemented by measuring $\bar{X}_1\bar{Z}_i$, followed by $\bar{X}_i$. The resulting state will be the same as $\bar{H}_i\ket{\bar{\psi}}$, except that the first and the $i$-th logical qubits are swapped, and the first qubit requires some logical Pauli correction depending on the measurement outcomes. Also, the $i$-th logical qubit will be left in the $\ket{\bar{+}}$ or $\ket{\bar{-}}$ state. The desired state can be obtained by applying a logical SWAP gate $\overline{\mathrm{SWAP}}_{1i}$ between the first and the $i$-th logical qubits.
    \item A logical phase gate $\bar{S}_i$ can be implemented by measuring $\bar{X}_1\bar{Y}_i$, followed by $\bar{Z}_i$. The resulting state will be the same as $\bar{S}_i\ket{\bar{\psi}}$ up to some logical Pauli correction, except that the first and the $i$-th logical qubits are swapped. By applying $\overline{\mathrm{SWAP}}_{1i}$, the desired state can be obtained.
    \item A logical CNOT gate $\overline{\mathrm{CNOT}}_{ij}$ can be implemented by measuring $\bar{X}_1\bar{X}_j$, followed by $\bar{Z}_i\bar{Z}_j$, followed by $\bar{X}_i$. The resulting state of all logical qubits after measurement is $\overline{\mathrm{SWAP}}_{1j}\overline{\mathrm{SWAP}}_{1i} \left(\ket{\bar{\pm}}\otimes\overline{\mathrm{CNOT}}_{ij}\ket{\bar{\psi}}\right)$, up to some logical Pauli correction. In other words, by applying $\overline{\mathrm{SWAP}}_{1j}$ followed by $\overline{\mathrm{SWAP}}_{1i}$, we can obtain the desired state $\overline{\mathrm{CNOT}}_{ij}\ket{\bar{\psi}}$ and the buffer qubit will be in the $\ket{\bar{+}}$ or $\ket{\bar{-}}$ state.
\end{enumerate}
With these set of logical Clifford gates, any logical Clifford circuit on a single code block can be implemented. 

\begin{figure}[!htbp]
	\centering
	\begin{subfigure}[b]{0.4\textwidth}
        \includegraphics[scale=1]{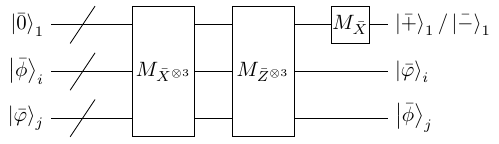}
        \captionsetup{justification=centering}
        \caption{}
    \end{subfigure}
    
    \vspace{1mm}
    \begin{subfigure}[b]{0.4\textwidth}
        \includegraphics[scale=1]{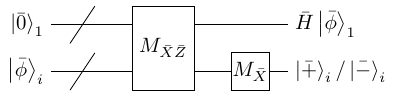}
        \captionsetup{justification=centering}
        \caption{}
    \end{subfigure}
    \vspace{1mm}
    \begin{subfigure}[b]{0.4\textwidth}
        \includegraphics[scale=1]{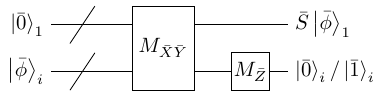}
        \captionsetup{justification=centering}
        \caption{}
    \end{subfigure}
    
    \begin{subfigure}[b]{0.4\textwidth}
        \includegraphics[scale=1]{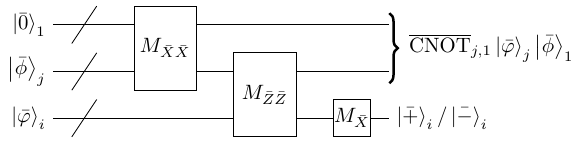}
        \captionsetup{justification=centering}
        \caption{}
    \end{subfigure}
	\caption{Implementations of basic logical Clifford gates through logical Pauli measurements (up to logical Pauli correction and logical qubit permutation) \cite{BZHJL15}. (a) A logical SWAP gate $\overline{\mathrm{SWAP}}_{ij}$. (b) A logical Hadamard gate $\bar{H}_i$. (c) A logical phase gate $\bar{S}_i$. (d) A logical CNOT gate $\overline{\mathrm{CNOT}}_{ij}$.}
	\label{fig:TPFTEC_basic_gates}%
\end{figure}

Note that in order to perform the computation fault-tolerantly, the ancilla states required for Steane FTM need to be prepared fault-tolerantly. This can be done efficiently using the scheme proposed in Ref. \cite{ZLB18}.

\section{Procedure for constructing a compatible symplectic basis and phase-type logical-level transversal gates for a self-dual CSS code}\label{app:python}

In this section, we describe a full procedure for verify whether a compatible symplectic basis of a self-dual CSS code exists, and if it does, construct a compatible symplectic basis of the code. The procedure can also construct phase-type logical-level transversal gates of the form $\bigotimes_{j=1}^{k} \bar{S}_j^{a_j}$ for any $(a_1,\dots,a_k)\in\{-1,1\}^k$ as a combination of physical $S$ and $S^\dagger$ gates. This procedure combines mathematical techniques used in the proofs of \cref{lem:alg,lem:basis_change,lem:stb_preserving_S} and \cref{thm:main1}.

For any self-dual CSS code $\mathcal{Q}$, our procedure takes as input a parity check matrix $H = \{\mathbf{g}_i\}$ (corresponding to stabilizer generators of $\mathcal{Q}$) and a set of coset representatives $A = \{\mathbf{h}_j\}$ (corresponding to one type of logical $X$ or logical $Z$ operators of $\mathcal{Q}$), and phase-type logical-level transversal gates we want to implement, i.e., $(a_1,\dots,a_k)\in\{-1,1\}^k$ of the logical gate $\bigotimes_{j=1}^{k} \bar{S}_j^{a_j}$. Note that the corresponding classical binary linear code and its dual are $\mathcal{D}=\langle \mathbf{g}_i,\mathbf{h}_j \rangle$ and $\mathcal{D}^\perp=\langle \mathbf{g}_i \rangle$. The procedure is as follows:
\begin{enumerate}
    \item Find $\mathbf{h}_j \in A$ such that $\mathbf{h}_j \cdot \mathbf{h}_j = 1$. If such $\mathbf{h}_j$ exists, CONTINUE (a compatible basis exists). If such $\mathbf{h}_j$ does not exist, STOP (a compatible basis does not exist).
    \item Apply \cref{alg:1} to $A = \{\mathbf{h}_j\}$ and obtain $\{\Bell_j^x\}$ and $\{\Bell_j^z\}$. Construct the first symplectic basis $\{(\bar{X}_j,\bar{Z}_j)\}_{j\in[k]}$ from $\{\Bell_j^x\}$ and $\{\Bell_j^z\}$. 
    \item Apply \cref{lem:basis_change} repeatedly on the first symplectic basis until a compatible symplectic basis is obtained.
    \item Find an operator $W$ which is a combination of physical $S$ and $S^\dagger$ that preserves the stabilizer group. $W$ can be found by solving the equation $H\mathbf{v} = \boldsymbol{\delta}$ using Gaussian elimination (see the proof of \cref{lem:stb_preserving_S} for more details).
    \item Compare the logical actions of $W$ from Step 5 with the logical actions specified by $(a_1,\dots,a_k)\in\{-1,1\}^k$, and apply the corresponding logical $Z$ operators to modify the logical actions of $W$. As a result, an operator $\bigotimes_{i=1}^{n} S_i^{b_i}$ where $(b_1,\dots,b_k)\in\{-1,1\}^k$ such that $\bigotimes_{j=1}^{k} \bar{S}_j^{a_j}=\bigotimes_{i=1}^{n} S_i^{b_i}$ can be obtained (see the proof of \cref{thm:main1}, $3 \Rightarrow 4$ for more details).
\end{enumerate}
A Python implementation of our procedure is available at \url{https://github.com/yugotakada/mlvtrans}.

\end{document}